\documentclass[submission,copyright,creativecommons]{eptcs}
\providecommand{\event}{PLACES 2013}

\usepackage{mathpartir}
\usepackage{stmaryrd}
\usepackage{amssymb}
\usepackage{empheq}
\usepackage{color,soul}
\usepackage{upgreek}
\usepackage{hyperref}
\usepackage{breakurl}
\usepackage{amsthm}
\usepackage{etoolbox}
\usepackage{tikz}
\usepackage{comment}
\usepackage{graphicx}

\usepackage{prooftree}

\newif\ifproofs
\proofstrue

\newcommand{\set}[1]{\{#1\}}
\newcommand{\subst}[2]{\{#1/#2\}}

\newcommand{\hypo}[1]{\text{\color{myred}\bfseries\upshape(#1)}}

\newcommand{\mkkeyword}[1]{\mathtt{\color{myblue}#1}}

\newcommand{\proofcase}[1]{\vskip1ex\noindent\fbox{#1}}
\newcommand{\rproofcase}[1]{\proofcase{\refrule{#1}}}
\newcommand{\eoe}{\hfill{$\blacksquare$}}

\newcommand{\defrule}[1]{%
  \hypertarget{rule:#1}{%
    \text{\small[\textsc{#1}]}%
  }%
}

\newcommand{\refrule}[1]{%
  \hyperlink{rule:#1}{%
    \text{\small[\textsc{#1}]}%
  }%
}

\definecolor{mymagenta}{rgb}{0.5,0,0.5}
\definecolor{mygreen}{rgb}{0,0.4,0}
\definecolor{myblue}{rgb}{0,0,0.6}
\definecolor{myred}{rgb}{0.4,0,0}
\definecolor{hlcolor}{rgb}{1,0.95,0}

\newcommand{\EndpointType}{\EndpointTypeT}
\newcommand{\EndpointTypeT}{T}
\newcommand{\EndpointTypeS}{S}

\newcommand{\var}{\varX}
\newcommand{\varX}{x}
\newcommand{\varY}{y}

\newcommand{\pvar}{\pvarX}
\newcommand{\pvarX}{X}
\newcommand{\pvarY}{Y}
\newcommand{\pvarZ}{Z}

\newcommand{\etvar}{\etvarA}
\newcommand{\etvarA}{\mathbf{t}}

\newcommand{\natset}{\mathbb{N}}

\newcommand{\ProcessSet}[1][]{
  \ifblank{#1}{
    \mathbb{P}
  }{
    \mathbb{P}^{[#1]}
  }
}
\newcommand{\FiniteProcessSet}{\ProcessSet[\mathsf{\color{myred}fin}]}

\newcommand{\Process}{\ProcessP}
\newcommand{\ProcessP}{P}
\newcommand{\ProcessQ}{Q}
\newcommand{\ProcessR}{R}

\newcommand{\Channel}{\ChannelA}
\newcommand{\ChannelA}{a}
\newcommand{\ChannelB}{b}
\newcommand{\ChannelC}{c}

\newcommand{\Name}{\NameU}
\newcommand{\NameU}{u}
\newcommand{\NameV}{v}

\newcommand{\Id}{\iota}

\newcommand{\mkts}[1]{\mathsf{\color{mygreen}#1}}
\newcommand{\timestamp}[1][\alpha]{\mkts{#1}}
\newcommand{\timestampM}{\timestamp[\alpha]}
\newcommand{\timestampN}{\timestamp[\beta]}

\newcommand{\Polarity}{\PolarityP}
\newcommand{\PolarityP}{p}
\newcommand{\PolarityQ}{q}

\newcommand{\EmptyEnv}{\emptyset}
\newcommand{\varassoc}[1]{
  \langle
  #1
  \rangle
}

\newcommand{\msg}[2]{
  #1
  \ifblank{#2}{}{\langle#2\rangle}
}
\newcommand{\imsg}[2]{
  #1
  \ifblank{#2}{}{(#2)}
}

\newcommand{\ep}[2]{
  #1^{#2}
}

\newcommand{\idle}{\mathbf{0}}
\newcommand{\new}[1]{
  (
  \nu#1
  )
}

\newcommand{\send}[2]{#1{!}\msg{}{#2}}
\newcommand{\receive}[2]{#1{?}\imsg{}{#2}}
\newcommand{\parop}{\mathbin{|}}

\newcommand{\Rec}[2][]{
  \mkkeyword{rec}
  \ifblank{#1}{}{
    {}^{[#1]}
  }
  \ifblank{#2}{}{
    ~#2
  }
}

\newcommand{\tint}{\mathit{int}}

\newcommand{\tmsg}[2]{
  #1
  \ifblank{#2}{}{
    #2
  }
}

\newcommand{\End}{\mkkeyword{end}}

\newcommand{\Action}[2][]{
  \ifblank{#1}{}{
    {\langle#1\rangle}
  }
  {#2}
}

\newcommand{\In}[3][]{\Action[#1]{?}\tmsg{#2}{#3}}
\newcommand{\Out}[3][]{\Action[#1]{!}\tmsg{#2}{#3}}
\newcommand{\trec}[1][]{
  \mu
  \ifblank{#1}{}{
    {}^{[#1]}
  }
}

\newcommand{\UEnv}{\Upgamma}
\newcommand{\WEnv}{\Upsigma}

\newcommand{\LEnv}{\Updelta}
\newcommand{\bind}[2]{#1 : #2}

\newcommand{\dom}{\mathsf{dom}}

\newcommand{\fpv}{\mathsf{fpv}}

\newcommand{\ftv}{\mathsf{ftv}}

\newcommand{\fn}{\mathsf{fn}}

\newcommand{\co}[1]{\overline{#1}}

\newcommand{\approximate}[2]{#1^{[#2]}}

\newcommand{\obligation}[2][]{
  \ifblank{#1}{
    \mathsf{ob}
  }{
    \mathsf{ob}_{#1}
  }
  (#2)
}

\newcommand{\capability}{\mathsf{cap}}

\newcommand{\eqdef}{\stackrel{\text{\tiny\upshape def}}{=}}

\newcommand{\red}{\rightarrow}

\newcommand{\nred}{\arrownot\red}

\newcommand{\asub}{\sqsubseteq}
\newcommand{\subt}{\leqslant}

\newcommand{\dualr}{\bowtie}

\newcommand{\wtp}[5][]{
  \ifblank{#2}{}{
    #2;
  }
  \ifblank{#3}{}{
    #3;
  }
  #4
  \vdash_{#1}
  #5
}

\makeatletter
\newtheorem*{rep@theorem}{\rep@title}
\newcommand{\newreptheorem}[2]{%
\newenvironment{rep#1}[1]{%
 \def\rep@title{#2 \ref{##1}}%
 \begin{rep@theorem}}%
 {\end{rep@theorem}}}
\makeatother

\theoremstyle{definition}
\newtheorem{definition}{Definition}[section]
\newtheorem{example}[definition]{Example}

\theoremstyle{plain}
\newtheorem{proposition}[definition]{Proposition}
\newtheorem{lemma}[definition]{Lemma}
\newtheorem{theorem}[definition]{Theorem}
\newreptheorem{theorem}{Theorem}
\newtheorem{corollary}[definition]{Corollary}

\theoremstyle{remark}

\ifproofs
\newenvironment{PROOF}{
\begin{proof}
}{
\end{proof}
}
\else

\fi

\title{From Lock Freedom to Progress Using Session Types}
\author{Luca Padovani
\institute{Dipartimento di Informatica, Universit\`a di Torino, Italy}
\email{luca.padovani@unito.it}
}
\def\titlerunning{From Lock Freedom to Progress Using Session Types}
\def\authorrunning{L. Padovani}
\begin{document}
\maketitle

\begin{abstract}
  Inspired by Kobayashi's type system for lock freedom, we define a
  behavioral type system for ensuring progress in a language of binary
  sessions.
  The key idea is to annotate actions in session types with priorities
  representing the urgency with which such actions must be performed
  and to verify that processes perform such actions with the required
  priority.
  Compared to related systems for session-based languages, the
  presented type system is relatively simpler and establishes progress
  for a wider range of processes.


\end{abstract}


\section{Introduction}
\label{sec:introduction}

A system has the progress property if it does not accumulate garbage
(messages that are produced and never consumed) and does not have dead
code (processes that wait for messages that are never produced).
For \emph{session-based} systems, where processes interact by means of
\emph{sessions} through disciplined interaction patterns described by
\emph{session types}, the type systems by Dezani-Ciancaglini \emph{et
  al.}
\cite{DezaniDeLiguoroYoshida07,CONCUR08,CoppoDezaniPadovaniYoshida13b}
guarantee that well-typed processes have progress. These type systems
analyze the dependencies between different (possibly interleaved)
sessions and establish progress if no circular dependency is found.
In a different line of work~\cite{Kobayashi02}, Kobayashi defines a
type system ensuring a lock-freedom property closely related to
progress. Despite the similarities between the notions of progress and
lock-freedom, however, the type systems
in~\cite{DezaniDeLiguoroYoshida07,CONCUR08,CoppoDezaniPadovaniYoshida13b}
and the one in~\cite{Kobayashi02} are difficult to compare, because of
several major differences in both processes and types. In particular,
the type systems in~\cite{CONCUR08,CoppoDezaniPadovaniYoshida13b} are
defined for an asynchronous language with a native notion of session,
while Kobayashi's type system is defined for a basic variant of the
synchronous, pure $\pi$-calculus.

The natural approach for comparing these analysis techniques would
require compiling a (well-typed) ``source'' session-based process into
a ``target'' $\pi$-calculus process, and then using Kobayashi's type
system for reasoning on progress of the source process in terms of
lock-freedom of the target one.
The problem of such compilation schemes
(see~\cite{DardhaGiachinoSangiorgi12} for an example) is that they
produce target processes in which the communication topology is
significantly more complex than that of the corresponding source ones
because of explicit continuation channel passing and encoding of
recursion. The net effect is that many well-typed source processes
become ill-typed according to~\cite{Kobayashi02}.
In this work we put forward a different approach: we lift the
technique underlying Kobayashi's type system to a session type system
for reasoning directly on the progress properties of processes.
The results are very promising, because the type system we obtain is
simpler than the ones defined
in~\cite{DezaniDeLiguoroYoshida07,CONCUR08,CoppoDezaniPadovaniYoshida13b}
and at the same time is capable of proving progress for a wider range
of processes. As a welcome side effect, the structure given by
sessions allows us to simplify some technical aspects of Kobayashi's
original type system as well.

\newcommand{\evx}[2]{
  \ifblank{#2}{
    {\color{mygreen}\timestamp[#1]}
  }{
    {\color{mygreen}\timestamp[#1]_{\mathsf{#2}}}
  }
}
\newcommand{\ev}[2][\alpha]{\evx{#1}{#2}}
\newcommand{\evA}{\ev[\alpha]{}}
\newcommand{\evB}{\ev[\beta]{}}
\newcommand{\evC}{\ev[\gamma]{}}
\newcommand{\evD}{\ev[\delta]{}}

To sketch the key ideas of Kobayashi's type system applied to
sessions, consider the term
\begin{equation}
\label{eq:mutual}
\receive{\ep{a}{+}}{\var}.\send{\ep{b}{-}}{4}
\parop
\receive{\ep{b}{+}}{\varY}.\send{\ep{a}{-}}{3}
\end{equation}
which represents the parallel composition of two processes that
communicate through two distinct sessions named $a$ and $b$. Each
session is accessed via its two endpoints, which we represent as the
name of the session decorated with a \emph{polarity} $+$ or $-$, along
the lines of~\cite{GayHole05}. We say that $\ep{a}-$ is the peer of
$\ep{a}+$, and vice versa. In \eqref{eq:mutual}, the process on the
left hand side of $\parop$ is waiting for a message from endpoint
$\ep{a}{+}$, after which it sends $4$ over endpoint $\ep{b}-$. The
process on the right hand side of $\parop$ instead is waiting for a
message from endpoint $\ep{b}+$, after which it sends $3$ over
endpoint $\ep{a}{-}$.
Notice that the message that is supposed to be received from $\ep{a}+$
is the one sent over $\ep{a}-$, and the message that is supposed to be
received from $\ep{b}+$ is the one sent over $\ep{b}-$. Clearly, as
each send operation is guarded by a receive, the term denotes a
process without progress. In particular, there is a circular
dependency between the actions pertaining the two sessions $a$ and
$b$.

The mechanism used for detecting these circular dependencies consists
in associating each action with an ordered pair of
\emph{priorities}. For instance, the receive action on $\ep{a}+$ would
be associated with the pair $\langle\evA, \evB\rangle$ where the first
component ($\evA$) measures the urgency to perform the action by the
process using $\ep{a}+$ and the second component ($\evB$) measures the
urgency to perform the complementary (send) action by the process
using the peer endpoint $\ep{a}-$. Because $\ep{a}-$ is the peer of
$\ep{a}+$, it is understood that such send action will be associated
with a pair that has exactly the same two components as $\langle\evA,
\evB\rangle$, but in reverse order, namely $\langle\evB, \evA\rangle$.
Similarly, the two actions on $\ep{b}+$ and $\ep{b}-$ will be
associated with two pairs $\langle\evC,\evD\rangle$ and
$\langle\evD,\evC\rangle$.
In the example, the two parallel processes are performing the receive
actions on the endpoints $\ep{a}+$ and $\ep{b}+$ first, therefore
complying with their respective duties no matter of how high the
priorities $\evA$ and $\evC$ are.
The send operation on $\ep{b}-$, on the other hand, is guarded by the
receive action on $\ep{a}+$ and will not be performed until this
action is completed, namely until a message is sent over endpoint
$\ep{a}-$. So, the left subprocess is complying with its duty to
perform the send action on $\ep{b}-$ with priority $\evD$ provided
that such priority is lower than that ($\evB$) to perform the action
on $\ep{a}-$.
At the same time, by looking at the right subprocess, we deduce that
such process is complying with its duty to perform the action on
$\ep{a}-$ with priority $\evB$ provided that $\evB$ is lower than
$\evD$ (the priority associated with the action on $\ep{b}-$).
Overall, we realize that the two constraints ``$\evD$ lower than
$\evB$'' and ``$\evB$ lower than $\evD$" are not simultaneously
satisfiable for any choice of $\evB$ and $\evD$, which is consistent
with the fact that the system makes no progress.

An even simpler example of process without progress is
\begin{equation}
\label{eq:self}
  \receive{\ep{a}+}{\var}.\send{\ep{a}-}{\var}
\end{equation}
where the input action on $\ep{a}+$ guards the very send action that
should synchronize with it. If we respectively associate the actions
on $\ep{a}+$ and $\ep{a}-$ with the pairs of priorities
$\langle\ev{},\evB{}\rangle$ and $\langle\evB{},\ev{}\rangle$ we see
that the structure of the process gives rise to the constraint
``$\evB$ lower than $\evB$'', which is clearly unsatisfiable.

In summary, the type system that we are going to present relies on
pairs of priorities associated with each action in the system, and
verifies whether the relations originating between these priorities as
determined by the structure of processes are satisfiable. If this is
the case, it can be shown that the system has progress.
All it remains to understand is the role played by session types. In
fact, in all the examples above we have associated priorities with
actions occurring within processes. Since a session type system
determines a one-to-one correspondence between actions occurring in
processes and actions occurring in session types, we let such pairs of
priorities be part of the session types themselves. For instance, the
left process in~\eqref{eq:mutual} would be well typed in an
environment with the associations $\ep{a}+ :
\In[\evA,\evB]{\tint}.\End$ and $\ep{b}- :
\Out[\evD,\evC]{\tint}.\End$ provided that ``$\evB$ is lower than
$\evD$''.

We continue the exposition by defining a calculus of binary sessions
in Section~\ref{sec:language}.
We purposely use a minimal set of supported features to ease the
subsequent formal development but we will be a bit more liberal in the
examples. The type system is defined in Section~\ref{sec:types},
\ifproofs
which also includes the soundness proof.
\fi
Section~\ref{sec:extensions} discusses a few extensions that can be
accommodated with straightforward adjustments to the type system.
Section~\ref{sec:conclusions} discusses related work and concludes.



\section{Language}
\label{sec:language}

\paragraph{Syntax.}
We begin by fixing a few conventions:
we use $m$, $n$, $\dots$ to range over natural numbers;
we use $\ChannelA$, $\ChannelB$, $\dots$ to range over (countably
many) \emph{channels};
we use $\PolarityP$, $\PolarityQ$, $\dots$ to range over the
\emph{polarities} $+$ and $-$;
we define an involution $\co{\,\cdot\,}$ over polarities such that
$\co+ = {-}$;
\emph{endpoints} $\ep\Channel+$, $\ep\Channel-$, $\dots$ are channels
decorated with a polarity;
we use $\varX$, $\varY$, $\dots$ to range over (countably many)
\emph{variables};
we use $\NameU$, $\NameV$, $\dots$ to range over \emph{names}, which
are either variables or endpoints;
we use $\Id$, $\dots$ to range over \emph{indices}, which are either
natural numbers or $\infty$;
we let $\infty + 1 = \infty$ and we extend the usual total order $<$
over natural numbers to indices so that $n < \infty$ for every $n$;
we use $\pvarX$, $\pvarY$, $\dots$ to range over (countably many)
\emph{process variables};
we use $\ProcessP$, $\ProcessQ$, $\dots$ to range over processes.

\begin{table}
  \caption{\label{tab:syntax} Syntax of processes.\strut}
\framebox[\columnwidth]{
\begin{math}
\displaystyle
\begin{array}[t]{@{}c@{\qquad}c@{}}
\begin{array}[t]{@{}rcl@{\quad}l@{}}
  \Process & ::= & & \textbf{Process} \\
  &   & \idle & \text{(idle)} \\
  & | & \pvar & \text{(variable)} \\
  & | & \receive\NameU{\var}.\Process & \text{(input)} \\
  & | & \send\NameU\NameV.\Process & \text{(output)} \\
  & | & \Process \parop \ProcessQ & \text{(composition)} \\
  & | & \new\ChannelA \Process & \text{(session)} \\
  & | & \Rec[\Id]\pvar.\Process & \text{(recursion)} \\
\end{array}
\end{array}
\end{math}
}
\end{table}

The language we work with is a simple variant of the synchronous
$\pi$-calculus equipped with binary sessions.
Each session takes place on a private channel $\Channel$, which is
represented as two \emph{peer endpoints} $\ep\Channel+$ and
$\ep\Channel-$ so that a message sent over one of the endpoints is
received from its peer.
The syntax of processes is defined by the grammar in
Table~\ref{tab:syntax} and briefly described in the following
paragraphs.
The term $\idle$ denotes the idle process, which performs no actions.
The term $\receive\NameU\var.\Process$ denotes a process that waits
for a message from endpoint $\NameU$, binds the message to the
variable $\var$, and then behaves as $\Process$.
The term $\send\NameU\NameV.\Process$ denotes a process that sends
message $\NameV$ over the endpoint $\NameU$ and then continues as
$\Process$.
In the prefixes $\receive\NameU\var$ and $\send\NameU\NameV$ we call
$\NameU$ the \emph{subject}.
The term $\ProcessP \parop \ProcessQ$ denotes the conventional
parallel composition of $\ProcessP$ and $\ProcessQ$.
The term $\new\Channel \Process$ denotes a session on channel
$\Channel$ that is private to $\Process$. Within $\Process$ the
session can be accessed through the two endpoints $\ep\Channel+$ and
$\ep\Channel-$.
Terms $\pvar$ and $\Rec[\Id]\pvar.\Process$ are used for building
recursive processes. The only unusual feature here is the index $\Id$
which, when finite, sets an upper bound to the number of unfoldings of
the recursive term.

A term $\receive\NameU\var.\Process$ binds the variable $\var$ in
$\Process$, a term $\new\Channel\Process$ binds the endpoints
$\ep\Channel+$ and $\ep\Channel-$ in $\Process$, and a term
$\Rec[\Id]\pvar.\Process$ binds the process variable $\pvar$ in
$\Process$. Then, $\fn(\Process)$ denotes the set of free names of a
generic process $\Process$. Similarly for $\fpv(\Process)$, but for
free process variables.
We sometimes write $\prod_{i=1}^n \Process_i$ for the composition
$\Process_1 \parop \cdots \parop \Process_n$ and
$\new{\tilde\Channel}\Process$ for
$\new{\Channel_1}\cdots\new{\Channel_n}\Process$.
We write $\ProcessSet[\Id]$ for the set of all processes such that
every $\Rec{}$ occurring in them has an index no greater than $\Id$
and we let $\FiniteProcessSet = \bigcup_{n\in\natset} \ProcessSet[n]$.
We say that $\Process$ is a \emph{user process} if $\Process \in
\ProcessSet[\infty] \setminus \FiniteProcessSet$. That is, user
processes only allow for unbounded unfoldings of recursions. The
remaining processes are only useful for proving soundness of the type
system.

\ifproofs
\begin{table}
  \caption{\label{tab:congruence} Structural congruence for processes.\strut}
\framebox[\columnwidth]{
\begin{math}
\displaystyle
\begin{array}[t]{@{}c@{}}
  \inferrule[\defrule{s-par 1}]{}{
    \idle \parop \Process \equiv \Process
  }
  \qquad
  \inferrule[\defrule{s-par 2}]{}{
    \ProcessP \parop \ProcessQ \equiv \ProcessQ \parop \ProcessP
  }
  \qquad
  \inferrule[\defrule{s-par 3}]{}{
    \ProcessP \parop (\ProcessQ \parop \ProcessR)
    \equiv
    (\ProcessP \parop \ProcessQ) \parop \ProcessR
  }
  \qquad
  \inferrule[\defrule{s-res 1}]{}{
    \new\ChannelA
    \new\ChannelB
    \Process
    \equiv
    \new\ChannelB
    \new\ChannelA
    \Process
  }
  \\\\
  \inferrule[\defrule{s-res 2}]{
    \ep\ChannelA+, \ep\ChannelA- \not\in \fn(\ProcessQ)
  }{
    \new\ChannelA \ProcessP \parop \ProcessQ
    \equiv
    \new\ChannelA (\ProcessP \parop \ProcessQ)
  }
\end{array}
\end{math}
}
\end{table}
\fi

\begin{table}
  \caption{\label{tab:reduction} Reduction of processes.\strut}
\framebox[\columnwidth]{
\begin{math}
\displaystyle
\begin{array}[t]{@{}c@{}}
  \inferrule[\defrule{r-comm}]{}{
    \send{\ep\ChannelA\PolarityP}{\ep\ChannelC\PolarityQ}.\ProcessP
    \parop
    \receive{\ep\ChannelA{\co\PolarityP}}{\var}.\ProcessQ
    \red
    \ProcessP
    \parop
    \ProcessQ\subst{\ep\ChannelC\PolarityQ}{\var}
  }
  \qquad
  \inferrule[\defrule{r-rec}]{}{
    \Rec[\Id+1]\pvar.\Process
    \red
    \Process\subst{\Rec[\Id]\pvar.\Process}{\pvar}
  }
  \\\\
  \inferrule[\defrule{r-res}]{
    \ProcessP
    \red
    \ProcessQ
  }{
    \new\ChannelA\ProcessP
    \red
    \new\ChannelA\ProcessQ
  }
  \qquad
  \inferrule[\defrule{r-par}]{
    \ProcessP \red \ProcessP'
  }{
    \ProcessP \parop \ProcessQ \red \ProcessP' \parop \ProcessQ
  }
  \qquad
  \inferrule[\defrule{r-struct}]{
    \ProcessP \equiv \ProcessP'
    \\
    \ProcessP' \red \ProcessQ'
    \\
    \ProcessQ' \equiv \ProcessQ
  }{
    \ProcessP \red \ProcessQ
  }
\end{array}
\end{math}
}
\end{table}

\paragraph{Reduction Semantics.}
The operational semantics of the calculus is expressed as usual as a
combination of a structural congruence, which rearranges equivalent
terms, and a reduction relation.
\ifproofs
Structural congruence is the least congruence that includes alpha
renaming of bound names and process variables and the laws in
Table~\ref{tab:congruence}.  It is basically the same as that of the
$\pi$-calculus, with the only exception of~\refrule{s-res 2} which
changes the scope of \emph{both endpoints} $\ep\Channel+$ and
$\ep\Channel-$ of a channel $\Channel$.
\fi
Reduction is the least relation defined by the rules in
Table~\ref{tab:reduction}. It includes two axioms for communication
\refrule{r-comm} and recursion unfolding \refrule{r-rec}, two
context rules \refrule{r-res} and \refrule{r-par}, and a rule for
reduction up to structural congruence \refrule{r-struct}.
Most rules are standard. In \refrule{r-comm}, a synchronization occurs
only between two endpoints of the same channel with dual polarities
and $\ProcessQ\subst{\ep\ChannelC\PolarityQ}{\var}$ denotes the
capture-avoiding substitution of endpoint $\ep\ChannelC\PolarityQ$ in
place of the free occurrences of $\var$ within $\ProcessQ$. Note, in
particular, that $(\new\ChannelC
\send{\ep\ChannelC+}\var.\idle)\subst{\ep\ChannelC-}{\var}$ is
undefined and that such a substitution is applicable only after alpha
renaming the bound channel $\ChannelC$ by means of structural
congruence.
Unfolding of recursions is allowed only when the index associated with
the recursive term is not zero. If different from $\infty$, the index
is decremented by the unfolding. The notation
$\ProcessP\subst{\ProcessQ}{\pvar}$ denotes the capture-avoiding
substitution of process $\ProcessQ$ in place of the free occurrences
of the process variable $\pvar$ in $\ProcessP$. For example,
$(\new\Channel \pvar)
\subst{\send{\ep\Channel+}{\ep\ChannelB+}}{\pvar}$ is undefined.

In the following we write $\red^*$ for the reflexive, transitive
closure of $\red$ and we say that $\Process$ is in \emph{normal form},
written $\Process \nred$, if there is no $\ProcessQ$ such that
$\Process \red \ProcessQ$.

\paragraph{Progress.}
We conclude this section with the formalization of the progress
property that we have alluded to in Section~\ref{sec:introduction}.

\begin{definition}[progress]
\label{def:progress}
We say that $\ProcessP$ has \emph{progress} if:
\begin{enumerate}
\item $\ProcessP \red^*
  \new{\tilde\Channel}(\send{\ep\ChannelA\PolarityP}{\ep\ChannelC\PolarityQ}.\Process' \parop
  \ProcessQ)$ implies $\ProcessQ \red^*
  \new{\tilde\ChannelB}(\receive{\ep\ChannelA{\co\PolarityP}}{\var}.\ProcessQ'
  \parop \ProcessR)$ where $\ChannelA$ does not occur in
  $\tilde\ChannelB$;

\item $\ProcessP \red^*
  \new{\tilde\ChannelA}(\receive{\ep\ChannelA\PolarityP}{\var}.\ProcessP' \parop
  \ProcessQ)$ implies $\ProcessQ \red^*
  \new{\tilde\ChannelB}(\send{\ep\ChannelA{\co\PolarityP}}{\ep\ChannelC\PolarityQ}.\ProcessQ'
  \parop \ProcessR)$ where $\ChannelA$ does not occur in
  $\tilde\ChannelB$.
\end{enumerate}
\end{definition}

Note in particular that our notion of progress differs from deadlock
freedom in the sense that it is not sufficient for a process to be
able to reduce in order for it to enjoy progress. For instance,
$\send{\ep\ChannelA+}{\ep\ChannelB-}.\idle \parop
\Rec[\infty]\pvar.\pvar$ does \emph{not} have progress even if it
admits an infinite sequence of reductions because the message
$\ep\ChannelB-$ is never consumed (no prefix
$\receive{\ep\ChannelA-}\var$ ever emerges).


\section{Session Types for Global Progress}
\label{sec:types}

\begin{table}
\caption{\label{tab:types} Syntax of session types.}
\framebox[\columnwidth]{
\begin{math}
\displaystyle
\begin{array}[t]{@{}rcl@{\quad}l@{}}
  \EndpointType & ::= & & \textbf{Session Type} \\
  &   & \End & \text{(termination)} \\
  & | & \etvar & \text{(type variable)} \\
  & | & \In[\timestampM,\timestampN]{}{\EndpointTypeS}.\EndpointTypeT &
  \text{(input)} \\
  & | & \Out[\timestampM,\timestampN]{}{\EndpointTypeS}.\EndpointTypeT &
  \text{(output)} \\
  & | & \trec[\Id]\etvar.\EndpointType & \text{(recursion)} \\
\end{array}
\end{math}
}
\end{table}

\paragraph{Definitions.}
We use $\EndpointTypeT$, $\EndpointTypeS$, $\dots$ to range over
\emph{session types}, $\etvar$, $\dots$ to range over (countably many)
\emph{session type variables}, and $\timestampM$, $\timestampN$,
$\dots$ to range over \emph{priorities}, which we concretely represent
as natural numbers with the interpretation that ``smaller number''
means ``higher priority'', $0$ denoting the highest priority.
The syntax of session types is described by the grammar in
Table~\ref{tab:types}.
The term $\End$ denotes an endpoint on which no further input/output
operation is possible.
The term
$\In[\timestampM,\timestampN]{}{\EndpointTypeS}.\EndpointTypeT$
denotes an endpoint that must be used with priority $\timestampM$ for
receiving a message of type $\EndpointTypeS$ and according to
$\EndpointTypeT$ afterwards.  Similarly, the term
$\Out[\timestampM,\timestampN]{}{\EndpointTypeS}.\EndpointTypeT$
denotes an endpoint that must be used with priority $\timestampM$ for
sending a message of type $\EndpointTypeS$ and according to
$\EndpointTypeT$ afterwards.
Following Kobayashi~\cite{Kobayashi02}, we sometimes call
$\timestampM$ \emph{obligation} and $\timestampN$ \emph{capability}:
the obligation $\timestampM$ associated with an action of an endpoint
expresses the \emph{duty} to perform the action with priority
$\timestampM$ by the process owning the endpoint;
the capability $\timestampN$ associated with an action of an endpoint
expresses the \emph{guarantee} that the corresponding complementary
action will be performed with priority $\timestampN$ by the process
owning the peer endpoint.
Terms $\etvar$ and $\trec[\Id]\etvar.\EndpointTypeT$ are used for
building recursive session types, as usual. Like in processes,
$\trec[\Id]$'s are decorated with an index $\Id$ denoting the number
of unfoldings allowed on this recursion, which is unbounded when $\Id
= \infty$. The only binder for session type variables is $\trec$, so the
notions of free and bound type variables are as expected. We write
$\ftv(\EndpointTypeT)$ for the set of free type variables of
$\EndpointTypeT$.

We restrict session types to the terms generated by the grammar in
Table~\ref{tab:types} that satisfy the following conditions:
\begin{itemize}
\item there are no subterms of the form
  $\trec\etvar_1\cdots\trec\etvar_n.\etvar_1$;

\item the terms $\EndpointTypeS$ in all prefixes
  $\In[\timestampM,\timestampN]{}\EndpointTypeS$ and
  $\Out[\timestampM,\timestampN]{}\EndpointTypeS$ are \emph{closed}.
\end{itemize}

The first condition ensures that session types are \emph{contractive}
and avoids meaningless terms such as $\trec\etvar.\etvar$. The second
condition ensures that session types are \emph{stratified} (a similar
constraint can be found
in~\cite{CastagnaDezaniGiachinoPadovani09,BarbaneraDeLiguoro10,Padovani12})
and is imposed to simplify the notion of duality
(Definition~\ref{def:duality}).

We take an \emph{iso-recursive} point of view and distinguish between
a recursive session type $\trec[\Id+1]\etvar.\EndpointTypeT$ and its
unfolding
$\EndpointTypeT\subst{\trec[\Id]\etvar.\EndpointTypeT}{\etvar}$, where
$\EndpointTypeT\subst{\EndpointTypeS}{\etvar}$ denotes the
capture-avoiding substitution of the free occurrences of $\etvar$ in
$\EndpointTypeT$ with $\EndpointTypeS$. Note that in the unfolding the
index $\Id + 1$ is decremented to $\Id$, unless $\Id = \infty$ in
which case it remains $\infty$.

A crucial notion of every theory of binary session types is that of
\emph{duality}, which relates the session types associated with the
peer endpoints of a session. Informally, two session types are dual of
each other if they specify complementary behaviors, whereby an input
action with a message of type $\EndpointTypeS$ in one session type is
matched by an output action with a message of the same type in the
dual session type.
Formally, we define duality as follows:

\begin{definition}[duality]
\label{def:duality}
\emph{Duality} is the least relation $\dualr$ defined by the rules
\[
\inferrule[\defrule{d-end}]{}{
  \End \dualr \End
}
\qquad
\inferrule[\defrule{d-var}]{}{
  \etvar \dualr \etvar
}
\qquad
\inferrule[\defrule{d-prefix}]{
  \EndpointTypeT \dualr \EndpointTypeT'
}{
  \In[\timestampM, \timestampN]{\EndpointTypeS}.\EndpointTypeT
  \dualr
  \Out[\timestampN, \timestampM]{\EndpointTypeS}.\EndpointTypeT'
}
\qquad
\inferrule[\defrule{d-rec}]{
  \EndpointTypeT \dualr \EndpointTypeS
}{
  \trec[\Id] \etvar.\EndpointTypeT
  \dualr
  \trec[\Id] \etvar.\EndpointTypeS
}
\qquad
\inferrule[\defrule{d-unfold}]{
  \EndpointTypeT \dualr \trec[\Id+1]\etvar.\EndpointTypeS
}{
  \EndpointTypeT \dualr \EndpointTypeS\subst{\trec[\Id]\etvar.\EndpointTypeS}{\etvar}
}
\]
plus the symmetric ones of \refrule{d-prefix} and
\refrule{d-unfold}.
\end{definition}

Rules~\refrule{d-end}, \refrule{d-var}, and \refrule{d-rec} are
standard from binary session type theories.
Rule~\refrule{d-prefix} is conventional except for the swapping of
priorities decorating the actions that we have just discussed.
Rule~\refrule{d-unfold} is necessary because our session types are
iso-recursive. In particular, thanks to this rule we can derive that
$\trec[\infty]\etvar.\In[\timestampM,\timestampN]{}\EndpointTypeS.\etvar
\dualr
\Out[\timestampN,\timestampM]{}\EndpointTypeS.\trec[\infty]\etvar.\Out[\timestampN,\timestampM]{}\EndpointTypeS.\etvar$
where, in the second session type, we have unfolded the recursion
once. This rule is necessary because the types associated with peer
endpoints will in general be unfolded independently.

The judgments of the type system have the form
$\wtp[\Id]{\WEnv}{\UEnv}{\LEnv}{\Process}$ where
\[
\begin{array}{c@{\qquad}c@{\qquad}c}
\begin{array}{r@{~}c@{~}l@{~}c@{~}l@{~}c@{~}l}
  \LEnv & ::= & \emptyset & \mid & \Name : \EndpointType & \mid &
  \LEnv, \LEnv \\
\end{array}
&
\begin{array}{r@{~}c@{~}l@{~}c@{~}l@{~}c@{~}l}
\UEnv & ::= & \emptyset & \mid & \pvar : \varassoc\LEnv & \mid &
\UEnv, \UEnv \\
\end{array}
&
\begin{array}{r@{~}c@{~}l@{~}c@{~}l@{~}c@{~}l}
\WEnv & ::= & \emptyset & \mid & \etvar : \timestamp & \mid & \WEnv, \WEnv \\
\end{array}
\end{array}
\]
respectively define the \emph{name environment} associating names
$\Name$ with session types $\EndpointType$, the \emph{process
  environment} $\UEnv$ associating process variables $\pvar$ with name
environments $\varassoc{\LEnv}$, and the \emph{type variable
  environment} $\WEnv$ associating type variables $\etvar$ with
priorities $\timestamp$.
For all the environments we let $\dom(\cdot)$ be the function that
returns their domain, we assume that composition through `,' is
defined only when the environments being composed have disjoint
domains, and we identify environments modulo commutativity and
associativity of `,' and neutrality of $\EmptyEnv$.
We also write $\UEnv \setminus \pvar$ for the restriction of $\UEnv$
to $\dom(\UEnv) \setminus \set\pvar$ and $\UEnv + \pvar :
\varassoc\LEnv$ for $(\UEnv \setminus \pvar), \pvar :
\varassoc\LEnv$. Similarly for $\WEnv$.

We will need to compute the obligation of a type, which measures the
urgency with which a value having that type must be used. Intuitively,
the obligation of a session type $\EndpointType$ is given by the
obligation of its topmost action. This leaves open the question as to
what is the priority of $\EndpointType$ if $\EndpointType$ has no
topmost action, in particular when $\EndpointType$ is $\End$ or a type
variable.
In the former case we should return a value that means ``no urgency at
all'', since an endpoint with type $\End$ should \emph{not} be
used. We will use the special value $\infty$ to this purpose.
In the latter case we need a type variable environment that keeps
track of the obligation associated with each type variable, as
determined by the recursive structure of the session type in which it
is bound.
More precisely, whenever $\ftv(\EndpointType) \subseteq \dom(\WEnv)$
we define $\obligation[\WEnv]{\EndpointType}$ as:
\[
\obligation[\WEnv]{\EndpointTypeT} \eqdef
\begin{cases}
  \infty & \text{if $\EndpointTypeT = \End$} \\
  \WEnv(\etvar) & \text{if $\EndpointTypeT = \etvar \in \dom(\WEnv)$}
  \\
  \timestampM & \text{if $\EndpointTypeT =
    \In[\timestampM,\timestampN]{}\EndpointTypeS.\EndpointTypeT'$ or 
    $\EndpointTypeT =
    \Out[\timestampM,\timestampN]{}\EndpointTypeS.\EndpointTypeT'$} \\
  \obligation[\WEnv]{\EndpointTypeS} & \text{if $\EndpointTypeT =
    \trec[\Id]\etvar.\EndpointTypeS$} \\
\end{cases}
\]
Note that $\obligation[\WEnv]{\trec[\Id]\etvar.\EndpointTypeS}$ is
well defined because session types are contractive.

\begin{table}
\caption{\label{tab:type_system} Type rules for processes.\strut}
\framebox[\columnwidth]{
\begin{math}
\displaystyle
\begin{array}[t]{@{}c@{}}
  \inferrule[\defrule{t-idle}]{}{
    \wtp[\Id]{\WEnv}{\UEnv}{
      \EmptyEnv
    }{
      \idle 
    }
  }
  \qquad
  \inferrule[\defrule{t-var}]{}{
    \wtp[\Id]{
      \WEnv, \tilde\etvar : \tilde\timestamp
    }{
      \UEnv,
      \pvar : \varassoc{
        \bind{\tilde\Name}{\tilde\etvar}
      }
    }{
      \bind{\tilde\Name}{\tilde\etvar}
    }{
      \pvar
    }
  }
  \\\\
  \inferrule[\defrule{t-input}]{
    \wtp[\Id]{\WEnv}{\UEnv}{
      \LEnv,
      \bind{\Name}{\EndpointTypeT},
      \bind{\var}{\EndpointTypeS}
    }{
      \Process
    }
    \\
    \forall\NameV\in\dom(\LEnv):
    \timestampN < \obligation[\WEnv]{\LEnv(\NameV)}
  }{
    \wtp[\Id]{\WEnv}{\UEnv}{
      \LEnv,
      \bind{\Name}{
        \In[\timestampM,\timestampN]{}{\EndpointTypeS}.\EndpointTypeT
      }
    }{
      \receive{\Name}{\var}.\Process
    }
  }
  \qquad
  \inferrule[\defrule{t-par}]{
    \wtp[\Id]{\WEnv}{\UEnv}{
      \LEnv_1
    }{
      \ProcessP
    }
    \\
    \wtp[\Id]{\WEnv}{\UEnv}{
      \LEnv_2
    }{
      \ProcessQ
    }
  }{
    \wtp[\Id]{\WEnv}{\UEnv}{
      \LEnv_1, \LEnv_2
    }{
      \ProcessP \parop \ProcessQ
    }
  }
  \\\\
  \inferrule[\defrule{t-output}]{
    \wtp[\Id]{\WEnv}{\UEnv}{
      \LEnv,
      \bind{\NameU}{\EndpointTypeT}
    }{
      \Process
    }
    \\
    \timestampN < \obligation[\WEnv]{\EndpointTypeS}
    \\
    \forall\NameV\in\dom(\LEnv):
    \timestampN < \obligation[\WEnv]{\LEnv(\NameV)}
  }{
    \wtp[\Id]{\WEnv}{\UEnv}{
      \LEnv,
      \bind{\NameU}{
        \Out[\timestampM, \timestampN]{}{\EndpointTypeS}.\EndpointTypeT
      },
      \bind{\NameV}{\EndpointTypeS}
    }{
      \send{\NameU}{\NameV}.\Process
    }
  }
  \qquad
  \inferrule[\defrule{t-end}]{
    \wtp[\Id]{\WEnv}{\UEnv}{\LEnv}{\Process}
  }{
    \wtp[\Id]{\WEnv}{\UEnv}{
      \LEnv, \Name : \End
    }{\Process}
  }
  \\\\
  \inferrule[\defrule{t-rec}]{
    \wtp[\Id]{
      \WEnv
      +
      \tilde\etvar : \obligation[\WEnv]{\tilde\EndpointTypeT}
    }{
      \UEnv
      +
      \pvar :\varassoc{
        \bind{\tilde\Name}{\tilde\etvar}
      }
    }{
      \bind{\tilde\Name}{\tilde\EndpointTypeT}
    }{
      \Process
    }
    \\
    \Id' \leq \Id
  }{
    \wtp[\Id]{\WEnv}{\UEnv}{
      \bind{\tilde\Name}{\trec[\Id']\tilde\etvar.\tilde\EndpointTypeT}
    }{
      \Rec[\Id']\pvar.\Process
    }
  }
  \qquad
  \inferrule[\defrule{t-session}]{
    \wtp[\Id]{\WEnv}{\UEnv}{
      \LEnv,
      \bind{\ep\Channel+}{\EndpointTypeT},
      \bind{\ep\Channel-}{\EndpointTypeS}
    }{
      \Process
    }
    \\
    \EndpointTypeT \dualr \EndpointTypeS
  }{
    \wtp[\Id]{\WEnv}{\UEnv}{
      \LEnv
    }{
      \new{\Channel}\Process
    }
  }
\end{array}
\end{math}
}
\end{table}

In the following we will make abundant use of sequences. For example,
$\tilde\Name$ denotes a (possibly empty) sequence $\Name_1, \dots,
\Name_n$ of names. With some abuse of notation we also use sequences
for denoting environments. For example, we write $\tilde\etvar :
\tilde\timestamp$ for $\etvar_1 : \timestamp_1, \dots, \etvar_n :
\timestamp_n$ and $\tilde\Name :
\trec[\Id]\tilde\etvar.\tilde\EndpointType$ for $\Name_1 :
\trec[\Id]\etvar_1.\EndpointType_1, \dots, \Name_n :
\trec[\Id]\etvar_n.\EndpointType_n$.

The typing rules for processes are defined in
Table~\ref{tab:type_system}.
Rule~\refrule{t-idle} states that the idle process is well typed only
in the empty name environment. This is because endpoints are linear
entities and the ownership of an endpoint with a type different from
$\End$ imposes its use.

Rules~\refrule{t-input} and~\refrule{t-output} deal with
prefixes. They check that the process is entitled to receive/send a
message on the endpoint $\Name$ and does so with the required
priority. In \refrule{t-input}, the received message $\var$ becomes
part of the receiver's name environment, as the receiver has acquired
its ownership. In \refrule{t-output}, the sent message $\NameV$ is
removed from the sender's name environment because its ownership has
been transferred.
The premise $\timestampN < \obligation[\WEnv]{\LEnv(\NameV)}$ for
every $\NameV \in \dom(\LEnv)$ can be explained in this way: a process
of the form $\receive{\Name}{\var}.\Process$ blocks until a message is
received from endpoint $\Name$. So, while this process is complying
with its duty to use $\Name$ regardless of the obligation
$\timestampM$ associated with it, it is also postponing the use of any
endpoint in $\dom(\LEnv)$ until this synchronization takes place. The
capability $\timestampN$ gives information about the priority with
which the peer endpoint of $\Name$ will be used elsewhere in the
system. Therefore, the process is respectful of the priorities of the
endpoints in $\dom(\LEnv)$ if they are lower (hence numerically
greater) than $\timestampN$.
Three considerations:
first of all notice that, the reasoning excludes that the peer
endpoint of $\Name$ is in $\dom(\LEnv)$. If it were, its obligation
would be $\timestampN$ and the premise would require the unsatisfiable
constraint $\timestampN < \timestampN$. This allows us to rule out
configurations such as that exemplified in~\eqref{eq:self}.
Second, if a process has in the name environment an endpoint whose
type has highest priority (hence obligation 0), the process must use
such endpoint immediately. If the process has two or more endpoints
with highest priority, the only way for the process to be well typed
is to fork into as many different parallel subprocesses, each
immediately using one of the endpoints with highest priority.
Third, in the case of \refrule{t-output} it is also required that
$\timestampN$ be strictly smaller than the obligation associated with
the type of the sent message. This is because such message cannot be
used until it is received, namely until the send operation is
completed.

Rule~\refrule{t-par} splits the name environment and distributes its
content among the composed processes.

Rule~\refrule{t-session} deals with session restrictions and augments
the name environment in the restricted process with the two peer
endpoints of the session, which must be related by duality.

Rules~\refrule{t-rec} and~\refrule{t-var} deal with recursions. The
former verifies that the name environment consists of endpoints with a
recursive type, therefore imposing a correspondence between recursive
processes and recursive types. Then, it checks that the body of the
recursion is well typed where the type variable environment has been
augmented with the obligations associated with the recursive type
variables, the process environment has been augmented with the
association that specifies the valid name environment that is expected
whenever the recursion process variable is met, and the name
environment is updated by opening up the recursive types. The rule also
checks that the $\Id'$ indices in the recursive types and in the
recursive process do not exceed the bound $\Id$.
Finally, rule~\refrule{t-end} discards names from the name
environment, provided that these have type $\End$.


\paragraph{Basic Properties.}
Below we collect a few basic properties of the type system leading to
the subject reduction result. We begin with two standard substitution
results, one for processes and the other one for endpoints.

\ifproofs
\begin{lemma}[weakening]
\label{lem:weakening}
If $\wtp[\Id]{\WEnv}{\UEnv}{\LEnv}{\Process}$ and $\WEnv \subseteq
\WEnv'$ and $\UEnv \subseteq \UEnv'$, then
$\wtp[\Id]{\WEnv'}{\UEnv'}{\LEnv}{\Process}$.
\end{lemma}
\fi

\begin{lemma}[process substitution]
\label{lem:proc_subst}
Let \hypo1 $\wtp[\Id]{\WEnv', \tilde\etvar : \tilde\timestamp}{\UEnv',
  \pvar : \varassoc{\tilde\Name : \tilde\etvar}}{\LEnv,
  \tilde\Name:\tilde\EndpointTypeT}{\ProcessP}$ and \hypo2
$\wtp[\Id]{\WEnv}{\UEnv}{\tilde\Name:\tilde\EndpointTypeS}{\ProcessQ}$
where $\tilde\timestamp = \obligation[\WEnv]{\tilde\EndpointTypeS}$
and $\WEnv \subseteq \WEnv'$ and $\UEnv \subseteq \UEnv'$.  Then
$\wtp[\Id]{\WEnv'}{\UEnv'}{\LEnv, \tilde\Name :
  \tilde\EndpointTypeT\subst{\tilde\EndpointTypeS}{\tilde\etvar}}{\ProcessP\subst{\ProcessQ}{\pvar}}$.
\end{lemma}
\ifproofs
\begin{proof}
  By induction on the derivation of \hypo1 and by cases on the last
  rule applied. We only show a few interesting cases:

\proofcase{\refrule{t-var} when $\Process = \pvarX$}
Then $\LEnv = \EmptyEnv$, $\tilde\EndpointTypeT = \tilde\etvar$ and we
conclude from \hypo2 with an application of Lemma~\ref{lem:weakening}.

\rproofcase{t-input}
We deduce:
\begin{itemize}
\item $\LEnv, \tilde\Name : \tilde\EndpointTypeT = \LEnv', \Name :
  \In[\timestampM, \timestampN]{\EndpointTypeS}.\EndpointTypeT$;

\item $\ProcessP = \receive\Name\var.\ProcessP'$;

\item $\wtp[\Id]{\WEnv', \tilde\etvar : \tilde\timestamp}{\UEnv', \pvar :
  \varassoc{\tilde\Name : \tilde\etvar}}{\LEnv', \Name
    : \EndpointTypeT, \var : \EndpointTypeS}{\ProcessP'}$;

\item $\timestampN < \obligation[\WEnv', \tilde\etvar :
  \tilde\timestamp]{\LEnv'(\NameV)}$ for every $\NameV \in
  \dom(\LEnv')$.
\end{itemize}

We only consider the case in which $\Name \in \dom(\LEnv)$, the case
$\Name \in \tilde\Name$ being analogous. Then $\LEnv' = \LEnv'',
\tilde\Name : \tilde\EndpointTypeT$ for some $\LEnv''$.
By induction hypothesis we deduce $\wtp[\Id]{\WEnv'}{\UEnv'}{\LEnv''',
  \Name : \EndpointTypeT, \var
  : \EndpointTypeS}{\ProcessP'\subst{\ProcessQ}{\pvar}}$ where
$\LEnv''' = \LEnv'', \tilde\Name :
\tilde\EndpointTypeT\subst{\tilde\EndpointTypeS}{\tilde\etvar}$.
Because of the hypothesis $\tilde\timestamp =
\obligation[\WEnv]{\tilde\EndpointTypeS}$ we also have
$\obligation[\WEnv', \tilde\etvar : \tilde\timestamp]{\LEnv'(\NameV)}
= \obligation[\WEnv']{\LEnv'''(\NameV)}$ for every $\NameV \in
\dom(\LEnv') = \dom(\LEnv''')$.
We conclude $\wtp[\Id]{\WEnv'}{\UEnv'}{\LEnv, \tilde\NameU :
  \tilde\EndpointTypeT\subst{\tilde\EndpointTypeS}{\tilde\etvar}}{\ProcessP\subst{\ProcessQ}{\pvar}}$
with an application of \refrule{t-input}.

\rproofcase{t-rec}
We deduce:
\begin{itemize}
\item $\LEnv, \tilde\Name : \tilde\EndpointTypeT = \tilde\NameV :
  \trec[\Id']\tilde\etvar'.\tilde\EndpointTypeT'$;

\item $\ProcessP = \Rec[\Id']\pvarY.\ProcessP'$;

\item $\wtp[\Id]{(\WEnv', \tilde\etvar : \tilde\timestamp) + \tilde\etvar'
    : \obligation[\WEnv', \tilde\etvar :
      \tilde\timestamp]{\tilde\EndpointTypeT'}}{(\UEnv', \pvarX :
    \varassoc{\tilde\NameU : \tilde\etvar}) + \pvarY :
    \varassoc{\tilde\NameV : \tilde\etvar'}}{\tilde\NameV :
    \tilde\EndpointTypeT'}{\ProcessP'}$;

\item $\Id' \leq \Id$.
\end{itemize}

We only consider the case in which $\pvarX \ne \pvarY$ and
$\tilde\etvar \cap \tilde\etvar' = \emptyset$ and $\tilde\NameU =
\tilde\NameV$ (hence $\LEnv = \emptyset$).
Because of the hypothesis $\tilde\timestamp =
\obligation[\WEnv]{\EndpointTypeS}$ we know that
$\obligation[\WEnv',\tilde\etvar:\tilde\timestamp]{\tilde\EndpointTypeT'}
=
\obligation[\WEnv']{\tilde\EndpointTypeT'\subst{\tilde\EndpointTypeS}{\etvar}}$.
Therefore, by induction hypothesis we deduce $\wtp[\Id]{\WEnv' +
  \tilde\etvar' :
  \obligation[\WEnv']{\tilde\EndpointTypeT'\subst{\tilde\EndpointTypeS}{\tilde\etvar}}}{
  \UEnv' + \pvarY : \varassoc{\tilde\NameU :
    \tilde\etvar'}}{\tilde\NameU :
  \tilde\EndpointTypeT'\subst{\tilde\EndpointTypeS}{\tilde\etvar}}{\ProcessP'\subst{\ProcessQ}{\pvarX}}$.
We conclude $\wtp[\Id]{\WEnv'}{\UEnv'}{\tilde\NameU : (\trec[\Id]
  \tilde\etvar'.\tilde\EndpointTypeT')\subst{\tilde\EndpointTypeS}{\tilde\etvar}}{\ProcessP\subst{\ProcessQ}{\pvar}}$
with an application of \refrule{t-rec}.
\end{proof}
\fi

\begin{lemma}[value substitution]
\label{lem:value_subst}
Let $\wtp[\Id]{\WEnv}{\UEnv}{\LEnv, \var : \EndpointTypeS}{\Process}$ and
$\ep\ChannelC\PolarityQ \not\in \dom(\LEnv)$ and
$\Process\subst{\ep\ChannelC\PolarityQ}{\var}$ is defined. Then
$\wtp[\Id]{\WEnv}{\UEnv}{\LEnv, \ep\ChannelC\PolarityQ
  : \EndpointTypeS}{\Process\subst{\ep\ChannelC\PolarityQ}{\var}}$.
\end{lemma}

To prove subject reduction one must formulate it for processes which
possibly have free endpoints. In doing so, it is necessary to impose,
on the name environment used for typing such processes, that it enjoys
a basic form of balancing, whereby the peer endpoints of the same
session are associated with dual types. Formally:

\begin{definition}[balanced context]
\label{def:balanced}
We say that $\LEnv$ is \emph{balanced} if $\ep\Channel\Polarity,
\ep\Channel{\co\Polarity} \in \dom(\LEnv)$ implies
$\LEnv(\ep\Channel\Polarity) \dualr \LEnv(\ep\Channel{\co\Polarity})$.
\end{definition}

\ifproofs
It is also necessary to determine an accurate correspondence between
the name environment \emph{before} the reduction, and the name
environment \emph{after} the reduction. For this reason, we define a
reduction relation also for environments which takes into account the
possible changes that can occur to the types in its range: either a
recursive type is unfolded, or two corresponding actions from types
associated with peer endpoints annihilate each other as the result of
a communication.

\begin{definition}[context reduction]
\label{def:context_reduction}
\emph{Context reduction} is the least relation $\red$ defined by the
rules
\[
\inferrule{}{
  \ep\Channel\Polarity : \trec[\Id+1]\etvar.\EndpointTypeT
  \red
  \ep\Channel\Polarity : \EndpointTypeT\subst{\trec[\Id]\etvar.\EndpointTypeT}{\etvar}
}
\qquad
\inferrule{}{
  \ep\Channel\Polarity : \Out[\timestampM,
  \timestampN]{\EndpointTypeS}.\EndpointTypeT,
  \ep\Channel{\co\Polarity} : \In[\timestampN,
  \timestampM]{\EndpointTypeS}.\EndpointTypeT'
  \red
  \ep\Channel\Polarity : \EndpointTypeT,
  \ep\Channel{\co\Polarity} : \EndpointTypeT'
}
\]
and closed by context composition.
\end{definition}

Context reductions preserve balancing.

\begin{lemma}
  Let $\LEnv$ be balanced and $\LEnv \red \LEnv'$. Then $\LEnv'$ is
  also balanced.
\end{lemma}
\begin{proof}
  By considering the two cases corresponding to the two possible
  reductions that can occur to $\LEnv$
  (Definition~\ref{def:context_reduction}). We show one of them.
  Suppose $\LEnv = \LEnv'', \ep\ChannelA\Polarity : \trec[\Id+1]
  \etvar.\EndpointTypeT \red \LEnv'', \ep\ChannelA\Polarity
  : \EndpointTypeT\subst{\trec[\Id]\etvar.\EndpointTypeT}{\etvar} =
  \LEnv'$ and that $\ep\Channel{\co\Polarity} \in \dom(\LEnv'')$.
  From the hypothesis that $\LEnv$ is balanced we deduce
  $\LEnv''(\ep\Channel{\co\Polarity}) \dualr \trec[\Id+1]
  \etvar.\EndpointTypeT$.
  We conclude $\LEnv''(\ep\Channel{\co\Polarity})
  \dualr \EndpointTypeT\subst{\trec[\Id]\etvar.\EndpointTypeT}{\etvar}$
  by an application of \refrule{d-unfold}.
\end{proof}

The property that typing is preserved by structural congruence is
obvious.

\begin{lemma}
\label{lem:struct}
Let $\wtp[\Id]{\WEnv}{\UEnv}{\LEnv}{\ProcessP}$ and $\ProcessP \equiv
\ProcessQ$. Then $\wtp[\Id]{\WEnv}{\UEnv}{\LEnv}{\ProcessQ}$.
\end{lemma}
\begin{proof}
  Standard induction on $\ProcessP \equiv \ProcessQ$.
\end{proof}
\fi

\begin{theorem}[subject reduction]
\label{thm:sr}
Let $\wtp[\Id]{}{}{\LEnv}{\ProcessP}$ and $\LEnv$ balanced and
$\ProcessP \red \ProcessQ$.  Then $\wtp[\Id]{}{}{\LEnv'}{\ProcessQ}$
\ifproofs
for some $\LEnv'$ such that $\LEnv \red^* \LEnv'$.
\fi
\end{theorem}
\ifproofs
\begin{proof}
  By induction on the derivation of $\ProcessP \red \ProcessQ$ and by
  cases on the last rule applied. We only focus on the two base cases,
  the remaining ones follow by a simple induction argument and
  possibly Lemma~\ref{lem:struct}.

\rproofcase{r-comm}
Then $\ProcessP = 
    \send{\ep\ChannelA\PolarityP}{\ep\ChannelC\PolarityQ}.\ProcessP'
    \parop
    \receive{\ep\ChannelA{\co\PolarityP}}{\var}.\ProcessQ'
    \red
    \ProcessP'
    \parop
    \ProcessQ'\subst{\ep\ChannelC\PolarityQ}{\var} = \ProcessQ$.
    From the hypothesis $\wtp[\Id]{}{}{\LEnv}{\ProcessP}$ and
    \refrule{t-par} we deduce:
\begin{itemize}
\item $\LEnv = \LEnv_1, \LEnv_2$;

\item
  $\wtp[\Id]{}{}{\LEnv_1}{\send{\ep\ChannelA\PolarityP}{\ep\ChannelC\PolarityQ}.\ProcessP'}$;

\item
  $\wtp[\Id]{}{}{\LEnv_2}{\receive{\ep\ChannelA{\co\PolarityP}}{\var}.\ProcessQ'}$.
\end{itemize}

From \refrule{t-output} we deduce:
\begin{itemize}
\item $\LEnv_1 = \LEnv_1', \ep\ChannelA\PolarityP : \Out[\timestampM,
  \timestampN]{\EndpointTypeS}.\EndpointTypeT, \ep\ChannelC\PolarityQ
  : \EndpointTypeS$;

\item $\wtp[\Id]{}{}{\LEnv_1', \ep\ChannelA\PolarityP
    : \EndpointTypeT}{\ProcessP'}$.
\end{itemize}

From \refrule{t-input} we deduce:
\begin{itemize}
\item $\LEnv_2 = \LEnv_2', \ep\ChannelA{\co\PolarityP} :
  \In[\timestampM', \timestampN']{\EndpointTypeS'}.\EndpointTypeT'$;

\item $\wtp[\Id]{}{}{\LEnv_2', \ep\ChannelA{\co\PolarityP} : \EndpointTypeT',
    \var : \EndpointTypeS'}{\ProcessQ'}$.
\end{itemize}

From the hypothesis that $\LEnv$ is balanced we also deduce that
$\EndpointTypeS' = \EndpointTypeS$ and $\EndpointTypeT \dualr
\EndpointTypeT'$.
By definition of $\LEnv$ we know that $\ep\ChannelC\PolarityQ \not\in
\dom(\LEnv_2')$.
By Lemma~\ref{lem:value_subst} we obtain $\wtp[\Id]{}{}{\LEnv_2',
  \ep\ChannelA{\co\PolarityP} : \EndpointTypeT',
  \ep\ChannelC\PolarityQ
  : \EndpointTypeS}{\ProcessQ'\subst{\ep\ChannelC\PolarityQ}{\var}}$.
Let $\LEnv' = \LEnv_1', \ep\ChannelA\PolarityP : \EndpointTypeT,
\LEnv_2', \ep\ChannelA{\co\PolarityP} : \EndpointTypeT',
\ep\ChannelC\PolarityQ : \EndpointTypeS$ and observe that $\LEnv \red
\LEnv'$.
We conclude $\wtp[\Id]{}{}{\LEnv'}{\ProcessQ}$ with an application of
\refrule{t-par}.

\rproofcase{r-rec}
Then $\ProcessP = \Rec[\Id+1]\pvar.\ProcessP' \red
\ProcessP'\subst{\Rec[\Id]\pvar.\ProcessP'}{\pvar} = \ProcessQ$.
Because of \refrule{t-end} we can assume that $\LEnv$ does not contain
bindings for endpoints with type $\End$.  Under this assumption, from
the hypothesis $\wtp[\Id]{}{}{\LEnv}{\ProcessP}$ and \refrule{t-rec}
we deduce:
\begin{itemize}
\item $\LEnv = \tilde\Name : \trec[\Id+1]
  \tilde\etvar.\tilde\EndpointTypeT$;

\item $\wtp[\Id]{\tilde\etvar : \obligation{\tilde\EndpointTypeT}}{\pvar :
    \varassoc{\tilde\Name : \tilde\etvar}}{\tilde\Name :
    \tilde\EndpointTypeT}{\ProcessP'}$.
\end{itemize}

From the hypothesis $\wtp[\Id]{}{}{\LEnv}{\ProcessP}$ we also deduce
$\wtp[\Id]{}{}{\tilde\NameU : \trec[\Id]
  \tilde\etvar.\tilde\EndpointTypeT}{\Rec[\Id]\pvar.\ProcessP'}$.
Note that $\obligation{\tilde\EndpointTypeT} = \obligation{\trec[\Id]
  \tilde\etvar.\tilde\EndpointTypeT}$. 
Let $\LEnv' = \tilde\Name : \tilde\EndpointTypeT\subst{\trec[\Id]
  \tilde\etvar.\tilde\EndpointTypeT}{\tilde\etvar}$ and observe that
$\LEnv \red^* \LEnv'$.
We conclude $\wtp[\Id]{}{}{\LEnv'}{\ProcessQ}$ by Lemma~\ref{lem:proc_subst}.
\end{proof}
\fi


\paragraph{Roadmap to Soundness.}
We sketch the proof that the type system is sound, namely that every
well-typed process $\ProcessP$ enjoys the progress property. According
to Definition~\ref{def:progress}, this amounts to showing that for
every $\ProcessP'$ such that
\[
\ProcessP \red^* \ProcessP'
\]
every top-level prefix involving some endpoint $\ep\Channel\Polarity$
in $\ProcessP'$ is eventually consumed by a matching prefix involving
the peer endpoint $\ep\Channel{\co\Polarity}$.
In this respect, what is difficult to prove is the existence of a
reduction sequence starting from $\ProcessP'$ that eventually exposes
the matching prefix, because the peer endpoint
$\ep\Channel{\co\Polarity}$ may be guarded by a number of prefixes
involving other endpoints. In fact, in $\Process'$ the endpoint
$\ep\Channel{\co\Polarity}$ may also be ``in transit'' as a message
exchanged within other sessions, hence the soundness proof should in
principle follow all the delegations of $\ep\Channel{\co\Polarity}$
until $\ep\Channel{\co\Polarity}$ becomes the subject of another
top-level prefix.

Instead of attempting this, we follow a radically different
strategy. First of all, we observe that a well-typed process in normal
form cannot have top-level prefixes (Lemma~\ref{lem:stability}).
The idea then is to prolong the derivation from $\Process'$ to some
$\Process''$ such that
\[
\Process \red^* \Process' \red^* \Process'' \nred
\]
and to conclude that the top-level prefix with subject
$\ep\Channel\Polarity$ in $\Process'$ must have been consumed by a
matching prefix that has emerged along the reduction from $\Process'$
to $\Process''$. Unfortunately, it is not always possible to find such
a $\Process''$ because in general well-typed processes (like
$\Process'$) are not weakly normalizing. However, since $\Process'$ is
a residual of $\Process$ after a \emph{finite} number of reductions,
it is possible to find a \emph{finite approximation} $\ProcessQ \in
\FiniteProcessSet$ of $\ProcessP$ that reduces to a finite
approximation $\ProcessQ' \in \FiniteProcessSet$ of
$\ProcessP'$. Since any process in $\FiniteProcessSet$ can be shown to
be strongly normalizing
\ifproofs
(Corollary~\ref{cor:sn}),
\fi
then there exists a normal form $\ProcessQ''$ that approximates
$\ProcessP''$.
The strategy is summarized by the following diagram
\[
\begin{array}{cccccc}
  \ProcessP & \red^* & \ProcessP' & \red^* & \ProcessP'' & \\
  \rotatebox[origin=c]{90}{$\asub$} & &
  \rotatebox[origin=c]{90}{$\asub$} & &
  \rotatebox[origin=c]{90}{$\asub$} \\
  \ProcessQ & \red^* & \ProcessQ' & \red^* & \ProcessQ'' & \nred
\end{array}
\]
where $\asub$ denotes some approximation relation.
Note that $\ProcessP''$ is not, in general, in normal form. However,
we will define $\asub$ in such a way that a user process and its
approximation share the same structure, except that the approximation
has finite indices marking the $\Rec[\infty]{}$ terms. Therefore, if
$\ProcessQ''$ has no top-level prefix, then so does $\ProcessP''$.


\paragraph{Approximations.}
Intuitively we say that $\ProcessP$ approximates $\ProcessQ$ if
$\ProcessP$ and $\ProcessQ$ share the same overall structure, except
that every recursion in $\ProcessQ$ is capable of at least as many
unfoldings as the corresponding recursion in $\ProcessP$. We formalize
this notion by means of an order between processes:

\begin{definition}[approximation]
  The $\asub$ be the least pre-congruence over processes induced by
  the rule
\[
\inferrule{
  \Id \leq \Id'
  \\
  \ProcessP \asub \ProcessQ
}{
  \Rec[\Id]\pvar.\ProcessP
  \asub
  \Rec[\Id']\pvar.\ProcessQ
}
\]

We say that $\ProcessP$ \emph{approximates} $\ProcessQ$ if $\ProcessP
\asub \ProcessQ$.
\end{definition}

The following Proposition establishes a simulation result between a
process and its approximations. In particular, the reductions of the
approximated process include those of its approximations.

\begin{proposition}
\label{prop:approx_red}
Let $\ProcessP \red^* \ProcessP'$ and $\ProcessP \asub \ProcessQ$.
Then there exists $\ProcessQ'$ such that $\ProcessQ \red^* \ProcessQ'$
and $\ProcessP' \asub \ProcessQ'$.
\end{proposition}
\ifproofs
\begin{proof}
  An easy induction on the derivation of $\ProcessP \red^*
  \ProcessP'$, using the fact that $\ProcessQ$ in general allows more
  unfoldings of its own recursions compared to $\ProcessP$.
\end{proof}
\fi

Our strategy for proving soundness relies on the ability to compute
one particular approximation of an arbitrary user process $\Process$.

\begin{definition}[$\Id$-approximant]
  The \emph{$\Id$-approximant} of a user process $\Process$, written
  $\approximate{\Process}{\Id}$, is obtained by turning every
  $\Rec[\infty]{}$ in $\Process$ to a $\Rec[\Id]{}$.
  We similarly define the $\Id$-approximant
  $\approximate{\EndpointTypeT}{\Id}$ of a session type
  $\EndpointTypeT$.
\end{definition}

An essential assumption of the strategy is that each $\Id$-approximant
of a well-typed user process is itself well typed. Unfortunately, this
is not always the case and we must slightly restrict the class of
well-typed user processes for which we are able to prove progress
using this strategy.  To illustrate the issue, consider the user
process
\[
\Process \eqdef
\new\ChannelA
(
\send{\ep\ChannelA{+}}{3}.
\Rec[\infty]\pvarX.
\send{\ep\ChannelA{+}}{3}.
\pvarX
\parop
\Rec[\infty]\pvarY.
\receive{\ep\ChannelA{-}}{\varX}.
\pvarY
)
\]
which is well typed using the name environment $\ep\ChannelA+
: \EndpointTypeT, \ep\ChannelA- : \EndpointTypeS$ where
\[
\EndpointTypeT \eqdef
\Out[\timestampM,\timestampN]{}\tint.\trec[\infty]\etvar.\Out[\timestampM,\timestampN]{}\tint.\etvar
\qquad\text{and}\qquad
\EndpointTypeS \eqdef
\trec[\infty]\etvar.\In[\timestampN,\timestampM]{}\tint.\etvar
\]
In particular, observe that the type $\EndpointTypeT$ associated with
$\ep\ChannelA+$ has been unfolded to account for the fact that in
$\Process$ the endpoint $\ep\ChannelA+$ is used once outside of the
recursion. Consequently the proof of $\EndpointTypeT
\dualr \EndpointTypeS$ crucially relies on \refrule{d-unfold}
(Definition~\ref{def:duality}) for dealing with this unfolding.
Now, it is easy to see that no $n$-approximant of $\Process$ is well
typed. In particular, it is not the case that
$\approximate\EndpointTypeT{n} \dualr \approximate\EndpointTypeS{n}$
because \refrule{d-unfold} attempts to relate
$\approximate\EndpointTypeS{n}$ with the \emph{folding} of
$\approximate\EndpointTypeT{n}$, namely
$\trec[n+1]\etvar.\Out[\timestampM,\timestampN]\tint.\etvar$.
In this particular case one could find a more clever approximation of
$\Process$ where the leftmost $\Rec{}$ is assigned index $n$ and the
rightmost one index $n+1$. However, because several endpoints can be
used within the same recursion, it is possible to find other examples
where \emph{no index assignment} makes the process typable with
finite indices.

In general, typability of every approximant of $\Process$ is
guaranteed if $\Process$ is typable without ever using
\refrule{d-unfold} for relating dual session types. This is the case
if $\wtp[0]{}{}{}{\approximate\Process{0}}$.

\begin{proposition}
\label{prop:approx_wt}
Let $\Process$ be a user process such that
$\wtp[0]{}{}{}{\approximate\Process{0}}$. Then
$\wtp[\Id]{}{}{}{\approximate\Process\Id}$ for every $\Id$.
\end{proposition}
\ifproofs
\begin{proof}
  The derivation for $\wtp[\Id]{}{}{}{\approximate\Process\Id}$ can be
  obtained from that for $\wtp[\Id]{}{}{}{\Process}$ by replacing
  every index $\infty$ occurring in $\Rec[\infty]{}$'s and
  $\trec[\infty]$'s with $\Id$.
\end{proof}
\fi

Given any finite reduction of a user process, it is possible to find
an appropriate finite approximant that simulates the reduction.

\begin{proposition}
\label{prop:approx_finite}
Let $\Process$ be a user process and $\ProcessP \red^*
\ProcessP'$. Then $\approximate{\ProcessP}{n} \red^* \ProcessQ \asub
\ProcessP'$ for some $n$ and $\ProcessQ \in \FiniteProcessSet$.
\end{proposition}
\ifproofs
\begin{proof}
  Just let $n$ be the number of reductions in the derivation of
  $\ProcessP \red^* \ProcessP'$. Then it is possible to simulate the
  reduction $\ProcessP \red^* \ProcessP'$ starting from
  $\approximate{\ProcessP}{n}$ to reach some $\ProcessQ \asub
  \ProcessP'$.
\end{proof}
\fi


\ifproofs
\paragraph{Strong Normalization of Finite Approximants.}
Let us address the strong normalization property of the
$\FiniteProcessSet$ fragment of the calculus. While this result is
intuitively obvious because each recursion can be unfolded only
finitely many times, the formal proof requires a rather complex
``measure'' for processes that decreases at each reduction step.
As a first attempt, one might define the measure of a process
$\Process$ as the vector where the item at index $i$ is the number of
$\Rec[i]{}$ terms occurring in $\Process$. This measure does not take
into account the fact that recursions with the highest index may
increase in number, if they occur nested within other recursions. For
instance, we have:
\[
\Rec[3]\pvarX.(\Rec[6]\pvarY.\pvarY \parop \pvarX)
\red
\Rec[6]\pvarY.\pvarY \parop \Rec[2]\pvarX.(\Rec[6]\pvarY.\pvarY \parop \pvarX)
\]

The example shows that the potential multiplicity of a recursive term
should also depend on the indices of the recursions within which it is
nested. Above, since the $\Rec[6]{}$ term occurs with a $\Rec[3]{}$
one, 3 unguarded instances of the $\Rec[6]{}$ term can be generated
overall. But this is true in the example above only because the
outermost recursion binds exactly one occurrence of the $\pvarX$
variable. In general, recursion variables can occur non-linearly. For
instance, we have
\[
\Rec[3]\pvarX.(\ProcessP \parop \pvarX \parop \pvarX)
\red
\ProcessP
\parop
\Rec[2]\pvarX.(\ProcessP \parop \pvarX \parop \pvarX)
\parop
\Rec[2]\pvarX.(\ProcessP \parop \pvarX \parop \pvarX)
\]
where the eventual number of unguarded $\ProcessP$ terms is 5.
In essence, in computing the multiplicity of a term we must consider
not only the indices of the recursions within which it is nested, but
also the multiplicity of the process variables bound by such
recursions.

\newcommand{\wvars}[2]{\mathbf{V}_{#2}(#1)}
\newcommand{\wproc}{\mathbf{E}}

Formally, we define an auxiliary function $\wvars{-}{-}$ such that
$\wvars{\ProcessP}{\pvar}$ provides the \emph{measure} of $\pvar$ in
$\ProcessP$, namely the number of occurrences of $X$ in $\ProcessP$,
taking into account duplications caused by inner recursions:
\[
\begin{array}{r@{~}c@{~}l@{\qquad}l}
  \wvars{\pvar}{\pvar} & = & 1 \\
  \wvars{\pvarY}{\pvarX} & = & 0 & \text{if $\pvarX \ne \pvarY$} \\
  \wvars{\send\NameU\NameV.\ProcessP}{X} =
  \wvars{\receive\NameU\var.\ProcessP}{X} & = & \wvars{\ProcessP}{X} \\
  \wvars{\ProcessP \parop \ProcessQ}{X} & = &
  \wvars{\ProcessP}{X} + \wvars{\ProcessQ}{X} \\
  \wvars{\Rec[n]\pvarX.\Process}{\pvarX} & = & 0 \\
  \wvars{\Rec[n]\pvarY.\Process}{\pvarX} & = &
  \wvars{\Process}{X} \cdot \sum_{k=0}^{n-1} \wvars{\Process}{Y}^k
  & \text{if $\pvarX \ne \pvarY$}
\end{array}
\]

All equations but the last one are unremarkable. In order to compute
the measure of $\pvarX$ in a process $\Rec[n]\pvarY.\Process$, we
multiply the measure of $\pvarX$ in $\Process$ by the amount of
duplication that $\pvarX$ is subjected to in all the unfoldings of
$\Rec[n]\pvarY.\Process$. This is determined by the geometric
progression $\sum_{k=0}^{n-1} \wvars{\Process}{Y}^k$ which, by
convention, is 0 when $n = 0$. In particular, variables guarded by a
$\Rec[0]{}$ term do not count, which is consistent with the fact that
such terms do not reduce (see \refrule{r-rec} in
Table~\ref{tab:reduction}).

Once we know how to determine the measure of variables, the measure of
terms follows similarly:
\[
\begin{array}{r@{~}c@{~}l}
  \wproc(X) & = & 0 \\
  \wproc(\send\NameU\NameV.\ProcessP) =
  \wproc(\receive\NameU\var.\ProcessP) & = & 1 + \wproc(\ProcessP) \\
  \wproc(\ProcessP \parop \ProcessQ) & = & \wproc(\ProcessP) +
  \wproc(\ProcessQ) \\
  \wproc(\Rec[n]\pvarX.\Process) & = & (1 + \wproc(\Process)) \cdot
  \sum_{k=0}^{n-1} \wvars{\Process}{X}^k
\end{array}
\]

In computing $\wproc(\Process)$ we also take into account the prefixes
of $\Process$, which may cause reductions by means of
\refrule{r-comm}. The measure of a recursive term
$\Rec[n]\pvarX.\Process$ is 1 (given by the unfolding of the term)
plus the measure of $\Process$ (after the unfolding) multiplied by the
amount of duplication that $\pvar$ is subjected to in the body of the
recursion.
As before, summations are empty when $n = 0$. In particular
$\wproc(\Rec[0]\pvar.\Process) = 0$ for every $\pvar$ and $\Process$.

The following crucial lemma shows that our notion of measure is well
behaved with respect to process substitutions:

\begin{lemma}
\label{lem:weight_subst}
Let $\ProcessP, \ProcessQ \in \FiniteProcessSet$ and
$\ProcessP\subst\ProcessQ{X}$ be defined.  Then
$\wproc(\ProcessP\subst\ProcessQ{X}) = \wproc(\ProcessP) +
\wproc(\ProcessQ) \cdot \wvars{\ProcessP}{X}$.
\end{lemma}
\ifproofs
\begin{proof}
\newcommand{\mysubst}{\subst\ProcessQ{X}}
By induction on the structure of $\ProcessP$. We only prove a few
interesting cases.

\proofcase{$\ProcessP = \send\NameU\NameV.\ProcessP'$}
  We have:

\begin{tabular}{r@{~}c@{~}p{0.475\textwidth}@{\quad}l}
  $\wproc(\ProcessP\mysubst)$ & $=$ &
  $\wproc(\send\NameU\NameV.\ProcessP'\mysubst)$ & definition of
  substitution \\
  & $=$ & $1 + \wproc(\ProcessP'\mysubst)$ & definition of $\wproc(-)$
  \\
  & $=$ & $1 + \wproc(\ProcessP') + \wproc(\ProcessQ) \cdot
  \wvars{\ProcessP'}{X}$ & induction hypothesis \\
  & $=$ & $\wproc(\ProcessP) + \wproc(\ProcessQ) \cdot
  \wvars{\ProcessP}{X}$ & definition of $\wproc(-)$ and $\wvars{-}{-}$ \\

\end{tabular}

\proofcase{$\ProcessP = \ProcessP_1 \parop \ProcessP_2$}
  We have:

\begin{tabular}{r@{~}c@{~}p{0.475\textwidth}@{\quad}l}
  $\wproc(\ProcessP\mysubst)$ & $=$ &
  $\wproc(\ProcessP_1\mysubst) + \wproc(\ProcessP_2\mysubst)$ &
  definition of $\wproc(-)$ \\
  & $=$ & $\wproc(\ProcessP_1) + \wproc(\ProcessQ) \cdot
  \wvars{\ProcessP_1}{X} + \wproc(\ProcessP_2) + \wproc(\ProcessQ)
  \cdot
  \wvars{\ProcessP_2}{X}$ & induction hypothesis \\
  & $=$ & $\wproc(\ProcessP) + \wproc(\ProcessQ) \cdot
  \wvars{\ProcessP}{X}$
  & definition of $\wproc(-)$ and $\wvars{-}{-}$
\end{tabular}

\proofcase{$\ProcessP = \Rec[n]\pvarY.\ProcessP'$ when $\pvarX \ne \pvarY$}
  We have:

\begin{tabular}{r@{~}c@{~}p{0.475\textwidth}@{\quad}l}
  $\wproc(\ProcessP\mysubst)$ & $=$ &
  $\wproc(\Rec[n]\pvarY.\ProcessP'\mysubst)$ & definition of
  substitution
  \\
  & $=$ & $(1 + \wproc(\ProcessP'\mysubst)) \cdot \sum_{k=0}^{n-1}
  \wvars{\ProcessP'\mysubst}{Y}^k$ & definition of $\wproc(-)$ \\
  & $=$ & $(1 + \wproc(\ProcessP'\mysubst)) \cdot \sum_{k=0}^{n-1}
  \wvars{\ProcessP'}{Y}^k$ & because $Y \not\in \fpv(\ProcessQ)$ \\
  & $=$ & $(1 + \wproc(\ProcessP') + \wproc(\ProcessQ) \cdot
  \wvars{\ProcessP'}{X}) \cdot \sum_{k=0}^{n-1} \wvars{\ProcessP'}{Y}^k$ &
  induction hypothesis \\
  & $=$ & $(1 + \wproc(\ProcessP')) \cdot \sum_{k=0}^{n-1}
  \wvars{\ProcessP'}{Y}^k$
  \\
  & & 
  \qquad ${} + \wproc(\ProcessQ) \cdot
  \wvars{\ProcessP'}{X} \cdot \sum_{k=0}^{n-1} \wvars{\ProcessP'}{Y}^k$
  & distributivity \\
  & $=$ & $\wproc(\ProcessP) + \wproc(\ProcessQ) \cdot
  \wvars{\ProcessP}{X}$ & definition of $\wproc(-)$ and $\wvars{-}{-}$
\end{tabular}
\end{proof}
\fi

The main result of this section states that the measure of a process
decreases at each reduction step.

\begin{theorem}
  Let $\ProcessP \in \FiniteProcessSet$ and $\ProcessP \red
  \ProcessQ$. Then $\wproc(\ProcessQ) < \wproc(\ProcessP)$.
\end{theorem}
\ifproofs
\begin{proof}
  By induction on the derivation of $\ProcessP \red \ProcessQ$ and by
  cases on the last rule applied. Here we only show the two base
  cases, the others following by the inductive argument possibly using
  the fact that $\equiv$ and endpoint substitutions preserve the
  measure of processes and process variables.

\rproofcase{r-comm}
Then $\ProcessP =
\send{\ep\ChannelA\PolarityP}{\ep\ChannelC\PolarityQ}.\ProcessP' \parop
\receive{\ep\ChannelA{\co\PolarityP}}{\var}.\ProcessQ' \red
\ProcessP' \parop \ProcessQ'\subst{\ep\ChannelC\PolarityQ}{\var} =
\ProcessQ$.
We conclude:

\begin{tabular}{r@{~}c@{~}l@{\qquad}l}
$\wproc(\ProcessQ)$ & $=$ $\wproc(\ProcessP') +
\wproc(\ProcessQ'\subst{\ep\ChannelC\PolarityQ}{\var}) =
\wproc(\ProcessP) - 2$
\end{tabular}

\noindent
using the fact that endpoint substitutions do not alter the measure of
a process.

\rproofcase{r-rec}
Then $\ProcessP = \Rec[n+1]\pvar.\ProcessP' \red
\ProcessP'\subst{\Rec[n]\pvar.\ProcessP'}{\pvar} = \ProcessQ$.
  We derive:

\begin{tabular}{r@{~}c@{~}l@{\qquad}l}
  $\wproc(\ProcessQ)$ & $=$ &
  $\wproc(\ProcessP') + \wproc(\Rec[n]\pvarX.\ProcessP') \cdot
  \wvars{\ProcessP'}{X}$ & Lemma~\ref{lem:weight_subst} \\
  & $=$ & $\wproc(\ProcessP') + (1 + \wproc(\ProcessP')) \cdot
  \sum_{k=0}^{n-1} \wvars{\ProcessP'}{X}^k \cdot \wvars{\ProcessP'}{X}$
  & definition of $\wproc(-)$ \\
  & $=$ & $\wproc(\ProcessP') + (1 + \wproc(\ProcessP')) \cdot
  (\sum_{k=0}^{n} \wvars{\ProcessP'}{X}^k - 1)$
  & geometric progression \\
  & $=$ & $\wproc(\ProcessP') + (1 + \wproc(\ProcessP')) \cdot
  \sum_{k=0}^{n} \wvars{\ProcessP'}{X}^k - (1 + \wproc(\ProcessP'))$
  & distributivity \\
  & $=$ & $(1 + \wproc(\ProcessP')) \cdot
  \sum_{k=0}^{n} \wvars{\ProcessP'}{X}^k - 1$
  &  \\
  & $=$ & $\wproc(\ProcessP) - 1$ & definition of $\wproc(-)$
\end{tabular}

\end{proof}
\fi

\begin{corollary}
\label{cor:sn}
Let $\Process \in \FiniteProcessSet$. Then $\Process$ is strongly
normalizing.
\end{corollary}

\fi


\paragraph{Soundness Results.}
The last auxiliary result we need concerns the shape of well-typed
processes in normal form, which are proved to have no pending prefixes
at the top level.

\begin{lemma}
\label{lem:stability}
Let $\wtp[\Id]{}{}{}{\Process}$ and $\Process \nred$. Then $\Process
\equiv \new{\tilde\Channel} \prod_{i\in I}
\Rec[0]{\pvar_i}.\Process_i$.
\end{lemma}
\ifproofs
\begin{proof}
  Using the structural congruence rules of Table~\ref{tab:congruence}
  it is clear that, whenever $\Process \nred$, we have $\Process
  \equiv \new{\tilde\Channel}\Process'$ for some $\Process'$ such that
\[
  \Process'
  ~~=~~
  \prod_{k\in K} \Rec[0]{\pvar_k}.\ProcessP_k
  ~~\parop~~
  \prod_{i=1}^{m} \send{\ep{\Channel_i}{\Polarity_i}}{\NameU_i}.\ProcessQ_i
  ~~\parop~~
  \prod_{\mathclap{i=m+1}}^{n} \receive{\ep{\Channel_i}{\Polarity_i}}{\var_i}.\ProcessR_i
\]
We now prove that $n = 0$.
From the hypothesis $\wtp[\Id]{}{}{}{\Process}$ we deduce
$\wtp[\Id]{}{}{\LEnv}{\Process'}$ for some $\LEnv$ that is balanced.
Let $\LEnv(\ep{\Channel_i}{\Polarity_i}) = \EndpointType_i$ and
$\LEnv(\ep{\Channel_i}{\co\Polarity_i}) = \co{\EndpointType_i}$ for
every $1 \le i \le n$.
Let $\capability(\EndpointType)$ be the capability of the topmost
action in $\EndpointType$, defined similarly to
$\obligation{\EndpointType}$, and observe that $\EndpointTypeT
\dualr \EndpointTypeS$ implies $\capability(\EndpointTypeT) =
\obligation{\EndpointTypeS}$.
We now proceed to show that for every $1\le i\le n$ there exists $1
\le j \le n$ such that $\capability(\EndpointType_j) <
\capability(\EndpointType_i)$. This is enough to conclude $n = 0$
because each $\capability(\EndpointType_i)$ is finite.

Let $1 \le i \le n$. By \refrule{t-input} and \refrule{t-output} we
deduce that $\EndpointType_i$ must begin with either an input or an
output, so $\ep{\Channel_i}{\co\Polarity_i}$ cannot occur in any of
the $\ProcessP_k$ because the type of endpoints occurring in
$\Process_k$ must begin with a $\trec[0]$ by \refrule{t-rec}.
Also, if $1 \le i \le m$, then $\ep{\Channel_i}{\co\Polarity_i}$
cannot be any of the $\ep{\Channel_j}{\Polarity_j}$ for $m + 1 \le j
\le n$ and if $m+1 \le i \le n$, then
$\ep{\Channel_i}{\co\Polarity_i}$ cannot be any of the
$\ep{\Channel_j}{\Polarity_j}$ for $1 \le j \le m$ because $\Process'
\nred$.
Suppose that $\ep{\Channel_i}{\co\Polarity_i} \in \set{\Name_j} \cup
\fn(\ProcessQ_j)$ for some $1 \le j \le m$.
By \refrule{t-output} we deduce $\capability(\EndpointType_j) <
\obligation{\co{\EndpointType_i}} = \capability(\EndpointType_i)$.
Suppose that $\ep{\Channel_i}{\co\Polarity_i} \in \fn(\ProcessR_j)$
for some $m + 1 \le j \le n$.
By \refrule{t-input} we deduce $\capability(\EndpointType_j) <
\obligation{\co{\EndpointType_i}} = \capability(\EndpointType_i)$.
\end{proof}
\fi

We conclude with the main result.

\begin{theorem}
\label{thm:progress}
Every user process $\Process$ such that
$\wtp[0]{}{}{}\approximate\Process{0}$ has progress.
\end{theorem}
\ifproofs
\begin{proof}
  Consider a derivation of $\Process \red^* \Process'$ where
  $\Process' =
  \new{\tilde\ChannelA}(\send{\ep\ChannelA\PolarityP}{\ep\ChannelC\PolarityQ}.\ProcessP_1 \parop
  \ProcessP_2)$.
  By Proposition~\ref{prop:approx_finite} there exist $n$ and
  $\ProcessQ'$ such that $\approximate{\Process}{n} \red^* \ProcessQ'$
  where $\ProcessQ' \in \FiniteProcessSet$ and $\ProcessQ' \asub
  \ProcessP'$.
  By Corollary~\ref{cor:sn} there exists $\ProcessQ''$ such that
  $\ProcessQ' \red^* \ProcessQ'' \nred$.
  From the hypothesis $\wtp[0]{}{}{}{\approximate\Process{0}}$,
  Proposition~\ref{prop:approx_wt}, and Theorem~\ref{thm:sr} we deduce
  $\wtp[n]{}{}{}{\ProcessQ''}$.
  From Proposition~\ref{prop:approx_red} we deduce that there exists
  $\ProcessP''$ such that $\ProcessP' \red^* \ProcessP''$ and
  $\ProcessQ'' \asub \ProcessP''$.
  From Lemma~\ref{lem:stability} we deduce that $\ProcessQ''$ does not
  contain unguarded prefixes, hence the same holds for $\ProcessP''$.
  We conclude that $\ProcessP_2 \red^*
  \new{\tilde\ChannelB}(\receive{\ep\Channel{\co\Polarity}}{\var}.\ProcessP_2'
  \parop \ProcessQ) \red^* \ProcessP''$ where $\ChannelA$ does not
  occur in $\tilde\ChannelB$.
\end{proof}
\fi


\begin{example}[forwarder]
\label{ex:forwarder}
Consider the process $\Process \eqdef \Rec[\infty] \pvar.
\receive{\ep\ChannelA-}{\var}.  \send{\ep\ChannelB+}{\var}.  \pvar$
which repeatedly receives a message from endpoint $\ep\ChannelA-$ and
forwards it to endpoint $\ep\ChannelB+$. Below is a derivation showing
that $\approximate{\Process}{0}$ is well typed in an appropriate name
environment:
\[
\begin{prooftree}
  \[
  \[
  \[
  \justifies
  \wtp[0]{
    \WEnv
  }{
    \UEnv
  }{
    \ep\ChannelA- : \etvar,
    \ep\ChannelB+ : \etvar'
  }{
    \pvar
  }
  \using \refrule{t-var}
  \]
  \quad
  \evD < \obligation[\WEnv]{\EndpointTypeS}
  \quad
  \evD < \evA
  \justifies
  \wtp[0]{
    \WEnv
  }{
    \UEnv
  }{
    \ep\ChannelA- : \etvar,
    \ep\ChannelB+ :
    \Out[\evC,\evD]{}\EndpointTypeS.\etvar',
    \var : \EndpointTypeS
  }{
    \send{\ep\ChannelB+}{\var}.
    \pvar
  }
  \using \refrule{t-output}
  \]
  \quad
  \evB < \evC
  \justifies
  \wtp[0]{
    \WEnv
  }{
    \UEnv
  }{
    \ep\ChannelA- : \In[\evA,\evB]{}\EndpointTypeS.\etvar,
    \ep\ChannelB+ : \Out[\evC,\evD]{}\EndpointTypeS.\etvar'
  }{
    \receive{\ep\ChannelA-}{\var}.
    \send{\ep\ChannelB+}{\var}.
    \pvar
  }
  \using \refrule{t-input}
  \]
  \justifies
  \wtp[0]{\EmptyEnv}{\EmptyEnv}{
    \ep\ChannelA- : \trec[0]\etvar.\In[\evA,\evB]{}\EndpointTypeS.\etvar,
    \ep\ChannelB+ : \trec[0]\etvar'.\Out[\evC,\evD]{}\EndpointTypeS.\etvar'
  }{
    \Rec[0] \pvar.
    \receive{\ep\ChannelA-}{\var}.
    \send{\ep\ChannelB+}{\var}.
    \pvar
  }
  \using \refrule{t-rec}
\end{prooftree}
\]

In the derivation we let $\WEnv \eqdef \etvar : \evA, \etvar' :
\evC$ and $\UEnv \eqdef \pvar : \varassoc{ \ep\ChannelA- :
  \etvar, \ep\ChannelB+ : \etvar' }$. Note that the constraints over
priorities are satisfiable, taking for instance
$\obligation[\WEnv]{\EndpointTypeS} = \evA = \evC = 1$ and
$\evD = 0$.
The interested reader can then extend the derivation to show that
\[
  \new\ChannelA\new\ChannelB
  (\approximate{\Process}{0}
  \mid
   \Rec[0]\pvarY.
   \new\ChannelC
   \send{\ep\ChannelA+}{\ep\ChannelC+}.
   \pvarY
  \mid
   \Rec[0]\pvarZ.
   \receive{\ep\ChannelB-}{\varY}.
   \pvarZ)
\]
is well typed (for instance, by taking $\EndpointTypeS = \End$, the
environment $\ep\ChannelC+, \ep\ChannelC- : \End$ within the
restriction $\new\ChannelC$, and using \refrule{t-end} in two
strategic places to discharge these endpoints), concluding that any
process having this as $0$-approximant has progress
(Theorem~\ref{thm:progress}).
Incidentally, the same example also shows the importance of
associating \emph{two distinct priorities} to each action. If we were
associating one single priority to each action, which essentially
amounts to adding the constraints $\evA = \evB$ and
$\evC = \evD$, this process would be ill typed because of
the unsatisfiable chain of constraints $\evA = \evB <
\evC = \evD < \evA$.  \eoe
\end{example}



\section{Extensions}
\label{sec:extensions}

%
Both the calculus and the types can be easily extended to support
labeled messages and label-driven branching.
The type language can also be enriched with basic data types such as
numbers, boolean values, etc. These values are not subject to any
linearity constraint, so the $\obligation{\cdot}$ function can be
conservatively extended to basic types by returning $\infty$, meaning
that \refrule{t-output} does not require any constraint when sending
messages of such types.

%
For simplicity our calculus is synchronous, but the type system
applies with minimal changes also to asynchronous communication, which
is more relevant in practice. In particular, since in an asynchronous
communication model output operations are non-blocking, rule
\refrule{t-output} can avoid to enforce the sequentiality of the
action with respect to the use of other endpoints.

%
Subtyping for session types has been widely studied
in~\cite{GayHole05,CastagnaDezaniGiachinoPadovani09,Padovani11b}. The
decorations that are necessary for enforcing progress allow a natural
form of subtyping, in accordance with the interpretation that a
channel with type $\EndpointTypeT$ can be safely used where a channel
with type $\EndpointTypeS$ is expected if $\EndpointTypeT
\subt \EndpointTypeS$ ($\EndpointTypeT$ is a subtype of
$\EndpointTypeS$). Indeed, by looking at the typing rules, it is clear
that obligations always occur on the right hand side of priority
constraints, while capabilities always occur on the left hand side of
these constraints. This means that subtyping can be covariant on
capabilities and contravariant on obligations. More precisely, the
core rules of subtyping would be formulated like this:
\[
\inferrule[\defrule{s-input}]{
  \timestampM' \leq \timestampM
  \\
  \timestampN \leq \timestampN'
  \\
  \EndpointTypeS \subt \EndpointTypeS'
  \\
  \EndpointTypeT \subt \EndpointTypeS'
}{
  \In[\timestampM, \timestampN]{}\EndpointTypeS.\EndpointTypeT
  \subt
  \In[\timestampM', \timestampN']{}{\EndpointTypeS'}.\EndpointTypeT'
}
\qquad
\inferrule[\defrule{s-output}]{
  \timestampM' \leq \timestampM
  \\
  \timestampN \leq \timestampN'
  \\
  \EndpointTypeS' \subt \EndpointTypeS
  \\
  \EndpointTypeT \subt \EndpointTypeS'
}{
  \Out[\timestampM, \timestampN]{}\EndpointTypeS.\EndpointTypeT
  \subt
  \Out[\timestampM', \timestampN']{}{\EndpointTypeS'}.\EndpointTypeT'
}
\]

%
Most session type theories support shared channel types that can be
distributed non-linearly among processes. In~\cite{BonoPadovani12} it
was shown that shared channel types can be added with minimum effort
by introducing a simple asymmetry between \emph{service types}, which
have the form $\In[\timestampM, \infty]{}\EndpointTypeS$ and only
allow receiving messages of type $\EndpointTypeS$, and \emph{client
  types}, which have the form $\Out[\infty,
\timestampM]{}\EndpointTypeS$ and only allow sending messages of type
$\EndpointTypeS$.
Service endpoints must be used linearly like session endpoints, to
make sure that no message sent over a client endpoint is lost. On the
contrary, client endpoints can be safely shared between multiple
processes or even left unused. As a result of this asymmetry, service
endpoint types are given finite obligation and infinite capability
(meaning that the owner of a service endpoint must use the channel,
but is not guaranteed that it will receive any message from it), and
dually client endpoint types are given infinite obligation and finite
capability (meaning that the owner of a client endpoint may not use
the endpoint, but if it does then it has the guarantee that the
message will be eventually received).
Because there is no guarantee that a message is sent over a client
endpoint, the progress property (Definition~\ref{def:progress}) must
be relaxed by allowing processes guarded by input actions on service
endpoints.



\section{Concluding Remarks}
\label{sec:conclusions}

By adapting the type system for lock freedom described
in~\cite{Kobayashi02} we have obtained a static analysis technique for
ensuring progress in a calculus of sessions that is more fine grained
than those described
in~\cite{DezaniDeLiguoroYoshida07,CONCUR08,CoppoDezaniPadovaniYoshida13b}. For
instance, the process shown in Example~\ref{ex:forwarder} is ill typed
according to the type systems
in~\cite{DezaniDeLiguoroYoshida07,CONCUR08,CoppoDezaniPadovaniYoshida13b}
where it is not allowed to delegate a received channel.  The increased
precision of the approach presented here comes from associating pairs
of priorities with each action in a session type, while in
\cite{CONCUR08,CoppoDezaniPadovaniYoshida13b} there is just one
priority associated with the shared name on which the session is
initiated.
Following the ideas presented by Kobayashi~\cite{Kobayashi02} and
adapted to sessions in the present work, Vieira and
Vasconcelos~\cite{VieiraVasconcelos13} have defined a similar type
system using abstract events instead of priorities, where events
represent the temporal order with which actions should be
performed. Their soundness result proves a weaker notion of progress,
but it should be possible to strengthen it along the lines of
Definition~\ref{def:progress}.

The aforementioned works can be classified as adopting a bottom-up
approach, in the sense that they aim at verifying a global property
(progress) of a compound system by checking properties of the system's
constituents (the sessions).
Other works adopt a top-down approach whereby well-typed or
well-formed systems have progress by design. For example, Carbone and
Montesi~\cite{CarboneMontesi13} advocate the use of a global
programming model for describing systems of communicating processes
such that, when the model is \emph{projected} into the constituent
processes, their parallel composition is guaranteed to enjoy progress.
Caires and Pfenning~\cite{CairesPfenning10} and subsequently
Wadler~\cite{Wadler12} present type systems such that well-typed terms
are deadlock-free. The result follows from the fact that the type
system prevents the same process to interleave actions pertaining to
different sessions.

%
A weakness of the type system presented here is that the priority
constraints checked by rules~\refrule{t-input} and~\refrule{t-output}
imply the knowledge of \emph{every} endpoint used in the continuation
of a process that follows a blocking action. This is feasible as long
as processes are described as terms of an abstract calculus, but in a
concrete programming language, processes are typically decomposed into
functions, methods, objects, and modules. While type checking each of
these entities in isolation, the type checker has only a partial
knowledge about the possible continuations of the program, and of
which endpoints are going to be used therein. We think that, in order
for the approach to be applicable in practice, it is necessary to
further enrich the structure of types. We are currently investigating
this problem in a language with first-order functions and
communication primitives.

\paragraph{Acknowledgments.}
I am grateful to Ilaria Castellani, Joshua Guttman and Philip Wadler,
who encouraged me to reconsider (and eventually dismiss) the
inaccurate interpretation of obligations and capabilities as
timestamps that I used in an earlier version of this paper.


\providecommand{\doi}[1]{%
  \href{http://dx.doi.org/\detokenize{#1}}{\textsc{doi:} \nolinkurl{\detokenize{#1}}}%
}

\bibliographystyle{eptcs}
\bibliography{main}

\begin{thebibliography}{10}
\providecommand{\bibitemdeclare}[2]{}
\providecommand{\surnamestart}{}
\providecommand{\surnameend}{}
\providecommand{\urlprefix}{Available at }
\providecommand{\url}[1]{\texttt{#1}}
\providecommand{\href}[2]{\texttt{#2}}
\providecommand{\urlalt}[2]{\href{#1}{#2}}
\providecommand{\doi}[1]{doi:\urlalt{http://dx.doi.org/#1}{#1}}
\providecommand{\bibinfo}[2]{#2}

\bibitemdeclare{inproceedings}{BarbaneraDeLiguoro10}
\bibitem{BarbaneraDeLiguoro10}
\bibinfo{author}{Franco \surnamestart Barbanera\surnameend} \&
  \bibinfo{author}{Ugo \surnamestart de'Liguoro\surnameend}
  (\bibinfo{year}{2010}): \emph{\bibinfo{title}{Two notions of sub-behaviour
  for session-based client/server systems}}.
\newblock In: {\sl \bibinfo{booktitle}{Proceedings of PPDP'10}},
  \bibinfo{publisher}{ACM}, pp. \bibinfo{pages}{155--164},
  \doi{10.1145/1836089.1836109}.

\bibitemdeclare{inproceedings}{CONCUR08}
\bibitem{CONCUR08}
\bibinfo{author}{Lorenzo \surnamestart Bettini\surnameend},
  \bibinfo{author}{Mario \surnamestart Coppo\surnameend},
  \bibinfo{author}{Loris \surnamestart D'Antoni\surnameend},
  \bibinfo{author}{Marco~De \surnamestart Luca\surnameend},
  \bibinfo{author}{Mariangiola \surnamestart Dezani-Ciancaglini\surnameend} \&
  \bibinfo{author}{Nobuko \surnamestart Yoshida\surnameend}
  (\bibinfo{year}{2008}): \emph{\bibinfo{title}{Global Progress in Dynamically
  Interleaved Multiparty Sessions}}.
\newblock In: {\sl \bibinfo{booktitle}{Proceedings of CONCUR'08}},
  \bibinfo{series}{LNCS 5201}, pp. \bibinfo{pages}{418--433},
  \doi{10.1007/978-3-540-85361-9_33}.

\bibitemdeclare{article}{BonoPadovani12}
\bibitem{BonoPadovani12}
\bibinfo{author}{Viviana \surnamestart Bono\surnameend} \&
  \bibinfo{author}{Luca \surnamestart Padovani\surnameend}
  (\bibinfo{year}{2012}): \emph{\bibinfo{title}{{T}yping {C}opyless {M}essage
  {P}assing}}.
\newblock {\sl \bibinfo{journal}{Logical Methods in Computer Science}}
  \bibinfo{volume}{8}, pp. \bibinfo{pages}{1--50},
  \doi{10.2168/LMCS-8(1:17)2012}.

\bibitemdeclare{inproceedings}{CairesPfenning10}
\bibitem{CairesPfenning10}
\bibinfo{author}{Lu\'{\i}s \surnamestart Caires\surnameend} \&
  \bibinfo{author}{Frank \surnamestart Pfenning\surnameend}
  (\bibinfo{year}{2010}): \emph{\bibinfo{title}{Session Types as Intuitionistic
  Linear Propositions}}.
\newblock In: {\sl \bibinfo{booktitle}{Proceedings of CONCUR'10}},
  \bibinfo{series}{LNCS 6269}, pp. \bibinfo{pages}{222--236},
  \doi{10.1007/978-3-642-15375-4_16}.

\bibitemdeclare{inproceedings}{CarboneMontesi13}
\bibitem{CarboneMontesi13}
\bibinfo{author}{Marco \surnamestart Carbone\surnameend} \&
  \bibinfo{author}{Fabrizio \surnamestart Montesi\surnameend}
  (\bibinfo{year}{2013}): \emph{\bibinfo{title}{Deadlock-freedom-by-design:
  multiparty asynchronous global programming}}.
\newblock In: {\sl \bibinfo{booktitle}{Proceedings of POPL'13}},
  \bibinfo{publisher}{ACM}, pp. \bibinfo{pages}{263--274},
  \doi{10.1145/2429069.2429101}.

\bibitemdeclare{inproceedings}{CastagnaDezaniGiachinoPadovani09}
\bibitem{CastagnaDezaniGiachinoPadovani09}
\bibinfo{author}{Giuseppe \surnamestart Castagna\surnameend},
  \bibinfo{author}{Mariangiola \surnamestart Dezani-Ciancaglini\surnameend},
  \bibinfo{author}{Elena \surnamestart Giachino\surnameend} \&
  \bibinfo{author}{Luca \surnamestart Padovani\surnameend}
  (\bibinfo{year}{2009}): \emph{\bibinfo{title}{{F}oundations of {S}ession
  {T}ypes}}.
\newblock In: {\sl \bibinfo{booktitle}{Proceedings of PPDP'09}},
  \bibinfo{publisher}{ACM}, pp. \bibinfo{pages}{219--230},
  \doi{10.1145/1599410.1599437}.

\bibitemdeclare{inproceedings}{CoppoDezaniPadovaniYoshida13b}
\bibitem{CoppoDezaniPadovaniYoshida13b}
\bibinfo{author}{Mario \surnamestart Coppo\surnameend},
  \bibinfo{author}{Mariangiola \surnamestart Dezani-Ciancaglini\surnameend},
  \bibinfo{author}{Luca \surnamestart Padovani\surnameend} \&
  \bibinfo{author}{Nobuko \surnamestart Yoshida\surnameend}
  (\bibinfo{year}{2013}): \emph{\bibinfo{title}{{Inference of Global Progress
  Properties for Dynamically Interleaved Multiparty Sessions}}}.
\newblock In: {\sl \bibinfo{booktitle}{Proceedings COORDINATION'13}},
  \bibinfo{volume}{LNCS 7890}, \bibinfo{publisher}{Springer}, pp.
  \bibinfo{pages}{45--59}, \doi{10.1007/978-3-642-38493-6_4}.

\bibitemdeclare{inproceedings}{DardhaGiachinoSangiorgi12}
\bibitem{DardhaGiachinoSangiorgi12}
\bibinfo{author}{Ornela \surnamestart Dardha\surnameend},
  \bibinfo{author}{Elena \surnamestart Giachino\surnameend} \&
  \bibinfo{author}{Davide \surnamestart Sangiorgi\surnameend}
  (\bibinfo{year}{2012}): \emph{\bibinfo{title}{Session types revisited}}.
\newblock In: {\sl \bibinfo{booktitle}{Proceedings of PPDP'12}},
  \bibinfo{publisher}{ACM}, pp. \bibinfo{pages}{139--150},
  \doi{10.1145/2370776.2370794}.

\bibitemdeclare{inproceedings}{DezaniDeLiguoroYoshida07}
\bibitem{DezaniDeLiguoroYoshida07}
\bibinfo{author}{Mariangiola \surnamestart Dezani-Ciancaglini\surnameend},
  \bibinfo{author}{Ugo \surnamestart de'Liguoro\surnameend} \&
  \bibinfo{author}{Nobuko \surnamestart Yoshida\surnameend}
  (\bibinfo{year}{2008}): \emph{\bibinfo{title}{On Progress for Structured
  Communications}}.
\newblock In: {\sl \bibinfo{booktitle}{Proceedings of TGC'07}},
  \bibinfo{series}{LNCS 4912}, pp. \bibinfo{pages}{257--275},
  \doi{10.1007/978-3-540-78663-4_18}.

\bibitemdeclare{article}{GayHole05}
\bibitem{GayHole05}
\bibinfo{author}{Simon \surnamestart Gay\surnameend} \&
  \bibinfo{author}{Malcolm \surnamestart Hole\surnameend}
  (\bibinfo{year}{2005}): \emph{\bibinfo{title}{Subtyping for session types in
  the $\pi$-calculus}}.
\newblock {\sl \bibinfo{journal}{Acta Informatica}}
  \bibinfo{volume}{42}(\bibinfo{number}{2-3}), pp. \bibinfo{pages}{191--225},
  \doi{10.1007/s00236-005-0177-z}.

\bibitemdeclare{article}{Kobayashi02}
\bibitem{Kobayashi02}
\bibinfo{author}{Naoki \surnamestart Kobayashi\surnameend}
  (\bibinfo{year}{2002}): \emph{\bibinfo{title}{A Type System for Lock-Free
  Processes}}.
\newblock {\sl \bibinfo{journal}{Information and Computation}}
  \bibinfo{volume}{177}(\bibinfo{number}{2}), pp. \bibinfo{pages}{122--159},
  \doi{10.1006/inco.2002.3171}.

\bibitemdeclare{inproceedings}{Padovani11b}
\bibitem{Padovani11b}
\bibinfo{author}{Luca \surnamestart Padovani\surnameend}
  (\bibinfo{year}{2011}): \emph{\bibinfo{title}{{S}ession {T}ypes =
  {I}ntersection {T}ypes + {U}nion {T}ypes}}.
\newblock In: {\sl \bibinfo{booktitle}{Proceedings of ITRS'10}},
  \bibinfo{volume}{EPTCS 45}, pp. \bibinfo{pages}{71--89},
  \doi{10.4204/EPTCS.45.6}.

\bibitemdeclare{article}{Padovani12}
\bibitem{Padovani12}
\bibinfo{author}{Luca \surnamestart Padovani\surnameend}
  (\bibinfo{year}{2012}): \emph{\bibinfo{title}{{O}n {P}rojecting {P}rocesses
  into {S}ession {T}ypes}}.
\newblock {\sl \bibinfo{journal}{Mathematical Structures in Computer Science}}
  \bibinfo{volume}{22}, pp. \bibinfo{pages}{237--289},
  \doi{10.1017/S0960129511000405}.

\bibitemdeclare{inproceedings}{VieiraVasconcelos13}
\bibitem{VieiraVasconcelos13}
\bibinfo{author}{Hugo~Torres \surnamestart Vieira\surnameend} \&
  \bibinfo{author}{Vasco~Thudichum \surnamestart Vasconcelos\surnameend}
  (\bibinfo{year}{2013}): \emph{\bibinfo{title}{Typing Progress in
  Communication-Centred Systems}}.
\newblock In: {\sl \bibinfo{booktitle}{Proceedings of COORDINATION'13}},
  \bibinfo{series}{LNCS 7890}, \bibinfo{publisher}{Springer}, pp.
  \bibinfo{pages}{236--250}, \doi{10.1007/978-3-642-38493-6_17}.

\bibitemdeclare{inproceedings}{Wadler12}
\bibitem{Wadler12}
\bibinfo{author}{Philip \surnamestart Wadler\surnameend}
  (\bibinfo{year}{2012}): \emph{\bibinfo{title}{Propositions as sessions}}.
\newblock In: {\sl \bibinfo{booktitle}{Proceedings of ICFP'12}},
  \bibinfo{publisher}{ACM}, pp. \bibinfo{pages}{273--286},
  \doi{10.1145/2364527.2364568}.

\end{thebibliography}

\end{document}